\documentclass[final,3p]{elsarticle}

\usepackage[usenames]{color}
\usepackage{amssymb}
\usepackage{amsthm}

\def\cl{{\cal C}\!\ell}
\def\Det{{\rm Det}}
\def\T{{\rm T}}

\def\tr{{\rm tr}}
\def\mod{{\, \rm mod \,}}
\def\Mat{{\rm Mat}}
\def\R{{\mathbb R}}
\def\C{{\mathbb C}}
\def\H{{\mathbb H}}

\def\diag{{\rm diag}}
\def\Adj{{\rm Adj}}
\def\id{{\rm id}}
\def\even{{\rm even}}
\def\odd{{\rm odd}}
\def\cen{{\rm cen}}
\def\adj{{\rm adj}}
\newtheorem{theorem}{Theorem}

\newtheorem{lemma}{Lemma}
\newtheorem{definition}{Definition}
\def\la{\langle}
\def\ra{\rangle}
\def\va{\vartriangle}

\journal{Computational and Applied Mathematics}

\begin{document}

\begin{frontmatter}

\title{On computing the determinant, other characteristic polynomial coefficients, and inverse in Clifford algebras of arbitrary dimension}

\author[mymainaddress,mysecondaryaddress]{D. S. Shirokov}
\ead{dm.shirokov@gmail.com}

\address[mymainaddress]{HSE University, 101000 Moscow, Russia}
\address[mysecondaryaddress]{Institute for Information Transmission Problems of Russian Academy of Sciences, 127051 Moscow, Russia}

\begin{abstract}
In this paper, we solve the problem of computing the inverse in Clifford algebras of arbitrary dimension. We present basis-free formulas of different types (explicit and recursive) for the determinant, other characteristic polynomial coefficients, adjugate, and inverse in real Clifford algebras (or geometric algebras) over vector spaces of arbitrary dimension $n$. The formulas involve only the operations of multiplication, summation, and operations of conjugation without explicit use of matrix representation. We use methods of Clifford algebras (including the method of quaternion typification proposed by the author in previous papers and the method of operations of conjugation of special type presented in this paper) and generalizations of numerical methods of matrix theory (the Faddeev-LeVerrier algorithm based on the Cayley-Hamilton theorem; the method of calculating the characteristic polynomial coefficients using Bell polynomials) to the case of Clifford algebras in this paper. We present the construction of operations of conjugation of special type and study relations between these operations and the projection operations onto fixed subspaces of Clifford algebras. We use this construction in the analytical proof of formulas for the determinant, other characteristic polynomial coefficients, adjugate, and inverse in Clifford algebras. The basis-free formulas for the inverse give us basis-free solutions to linear algebraic equations, which are widely used in computer science, image and signal processing, physics, engineering, control theory, etc. The results of this paper can be used in symbolic computation.
\end{abstract}

\end{frontmatter}

\section{Introduction}
\label{sec:1}

The problem of computing the determinant and inverse in Clifford algebras (or geometric algebras, \cite{Lounesto}) $\cl_{p,q}$, $p+q=n$ is very important from the theoretical and practical points of view. The explicit (symbolic, basis-free) formulas for the inverse of Clifford algebra elements give us explicit formulas for the solutions to linear algebraic equations $AXB=C$ for known $A, B, C\in\cl_{p,q}$ and unknown $X\in\cl_{p,q}$. The results of this paper give us the basis-free solution to the Sylvester equation $AX+XB=C$ in the case of arbitrary dimension $n=p+q$, see \cite{ShirokovSylv} (see the cases $n\leq 3$ in \cite{acusSylv}). Note that the Sylvester equation and its particular case, the Lyapunov equation (with $B=A^H$), are widely used in different applications -- image processing, control theory, stability analysis, signal processing, model reduction, and many more.

Over the past years, several results were obtained on the problem of computing the determinant and inverse in Clifford algebras. The basis-free formulas for the inverse of Clifford algebra elements for the cases $n\leq 5$ were presented in \cite{LS, dadbeh, rudn, hitzer1} using different methods. For the case $n=6$, the explicit formula was presented for the first time in \cite{acus}. In \cite{hitzer2}, it was presented how to obtain the algebraic expression for the inverse in the case of arbitrary odd $n$, if we know the corresponding expression for the case of previous even $n-1$. In this paper, we generalize results to the case of arbitrary $n$ using different techniques. We present the basis-free formulas of different types (explicit and recursive) for the determinant, other characteristic polynomial coefficients, adjugate, and inverse in the case of arbitrary $n$ using only the operations of multiplication, summation, and operations of conjugation without using the corresponding matrix representations. The results of this paper can be used in symbolic computation (using different software \cite{AblMaple, AcusPack, Python, HitzerMatlab})\footnote{One of the anonymous reviewers noted that he implemented and tested (both recursive and explicit) algorithm and checked that it yields correct results for Clifford algebras up to dimension $n=11$ (using elements with random integer coefficients), and stated that explicit algorithm is much more efficient than recursive; he also computed (using optimized version of the formula, provided by Lemma 6 in the article) of symbolic expression for determinant for $\cl_{6,0}$ in expanded form in approximately one day, whereas symbolic matrix determinant computation (in expanded form) took more than four days.}\footnote{See also applications of the results of this paper in symbolic computation using the Mathematica package \cite{AcusPack}, https://github.com/ArturasAcus/GeometricAlgebra, and the Python package \cite{Python}, https://github.com/pygae/clifford/pull/373.}.

The paper is organized as follows.

In Section \ref{sec:2}, we propose a new construction of $[\log_2 n]+1$ operations of conjugation of special type in Clifford algebras $\cl_{p,q}$, $p+q=n$. The standard fundamental operations (the grade involution, the reversion, and their superposition, which is called the Clifford conjugation) are naturally included in this construction. We realize all other $2^{n+1}$ operations of conjugations as linear combinations of superpositions of the proposed $m=[\log_2 n]+1$ operations of conjugation $\va_1$, \ldots, $\va_m$ of special type. In some sense, the use of the proposed $[\log_2 n]+1$ operations of conjugation of special type is an alternative to the use of $n+1$ grade-negation operations (see the papers \cite{dadbeh, acus, hitzer1}). We study properties of the proposed operations and relations between these operations and the operations of projection onto fixed subspaces of Clifford algebra, especially the projection onto the subspace of grade $0$ (scalar part of the element), which is related to the trace -- one of the invariants (characteristic polynomial coefficients) of the element.

In Section \ref{sec:3}, we give an algebraic proof of several formulas for the functionals $N(U):\cl_{p,q} \to \cl^0_{p,q}\equiv \R$ of special form in Clifford algebras over vector spaces of dimension $n\leq 6$. We present some new formulas (for $n=4, 5$) and prove them analytically in this section. We prove that each of 92 formulas (20 formulas in the form of doublets and 72 formulas in the form of triplets) obtained by computer calculations in \cite{acus} for the case $n=6$ is equal to one of the two formulas (\ref{nu6}), (\ref{nu62}) presented in this paper using only two or three operations of conjugation. The main tools in presented analytical proofs are the properties of the operations of conjugation of special type $\va_1$, \ldots, $\va_m$ proposed in this paper and the method of quaternion typification in Clifford algebras proposed in the previous papers of the author \cite{quat, quat1, quat2}.

In Section \ref{sec:4}, we generalize the concept of characteristic polynomial coefficients (in particular, the trace and the determinant) to the case of real Clifford algebras. To introduce these concepts in real Clifford algebras, we use matrix representations of minimal dimensions of the complexified Clifford algebras $\C\otimes\cl_{p,q}$ as the matrix representations (of non-minimal dimension) of the corresponding real Clifford algebras $\cl_{p,q}\subset \C\otimes\cl_{p,q}$. Then we prove that these concepts do not depend on the choice of matrix representation and give alternative definitions of these concepts without using matrix representations and using only Clifford algebra operations. We present explicit formulas for the determinant, other characteristic polynomial coefficients, adjugate, and inverse in the case of arbitrary $n$ using only the operations of multiplication, summation, and operations of conjugation without explicit use of matrix representation. To obtain these results,we generalize several matrix methods (the Faddeev-LeVerrier algorithm based on the Cayley-Hamilton theorem; the method of calculating the characteristic polynomial coefficients using Bell polynomials) to the case of Clifford algebras. Also we use the properties of the operations $\va_1$, \ldots, $\va_m$ presented in Section \ref{sec:2}. The examples are given in the cases of small $n$.

\section{Operations of conjugation of special type in Clifford algebras and their properties}
\label{sec:2}

Let us consider the real Clifford algebra (or the geometric algebra) $\cl_{p,q}$, $p+q=n$, with the generators $e_a$, $a=1, \ldots, n$, and the identity element $e$. The generators satisfy the conditions
$$e_a e_b+e_b e_a=2\eta_{ab}e,\qquad a, b=1, \ldots, n,$$
where $\eta=||\eta_{ab}||=\diag(1,\ldots, 1, -1, \ldots, -1)$ is the diagonal matrix with its first $p$ entries equal to $1$ and the last $q$ entries equal to $-1$ on the diagonal. We call the subspace of $\cl_{p,q}$ of Clifford algebra elements, which are linear combinations of the basis elements $e_{a_1 \ldots a_k}:=e_{a_1}\cdots e_{a_k}$, $a_1 < a_2 < \cdots < a_k$, with multi-indices of length $k$, the subspace of grade $k$ and denote it by $\cl^k_{p,q}$, $k=0, 1, \ldots, n$. Elements of grade $0$ are identified with scalars $\cl^0_{p,q}\equiv \R$, $e\equiv 1$.

We denote the projection of an element $U\in\cl_{p,q}$ onto the subspace $\cl^k_{p,q}$ by $\la U\ra_k$ (or sometimes by $U_k$ to simplify notation). The operations of projection are linear:
\begin{eqnarray}
&&\la U +V \ra_k=\la U \ra_k+ \la V \ra_k,\qquad \la \lambda U \ra_k=\lambda \la U \ra_k,\qquad \lambda\in\R,\qquad U, V\in\cl_{p,q}.\label{po}
\end{eqnarray}
We denote the projection of $U\in\cl_{p,q}$ onto the center of a Clifford algebra
\begin{eqnarray}\label{cen}
\cen(\cl_{p,q})=\left\lbrace
\begin{array}{ll}
\cl^0_{p,q}, & \mbox{if $n$ is even,}\\
\cl^0_{p,q}\oplus\cl^n_{p,q}, & \mbox{if $n$ is odd,}
\end{array}
\right.
\end{eqnarray}
by $\la U \ra_{\cen}$. If $n$ is even, then $\la U \ra_{\cen}=\la U \ra_0$. If $n$ is odd, then $\la U\ra_{\cen}=\la U \ra_0+\la U\ra_n$.

\begin{lemma}\label{lem1} We have the following properties
\begin{eqnarray}
&&\la UV \ra_0=\la VU \ra_0,\qquad \mbox{for an arbitrary $n$;}\label{larao}\\
&&\la UV \ra_n=\la VU \ra_n,\qquad \mbox{for odd $n$.}
\end{eqnarray}
As a consequence, we get
\begin{eqnarray}
&&\la UVW \ra_0=\la VWU\ra_0=\la WUV \ra_0,\quad \la T^{-1}UT \ra_0=\la U\ra_0,\qquad \mbox{for an arbitrary $n$;}\label{lara02}\\
&&\la UVW \ra_n=\la VWU\ra_n=\la WUV \ra_n,\quad \la T^{-1}UT \ra_n=\la U\ra_n,\qquad \mbox{for odd $n$,}
\end{eqnarray}
for all $U, V, W\in\cl_{p,q}$ and $T\in\cl_{p,q}^\times$, where $\cl_{p,q}^\times$ is the group of all invertible elements of $\cl_{p,q}$.
\end{lemma}
\begin{proof} One can find the proof of the facts that $ \la [U, V] \ra_0=0$ in the case of arbitrary $n$ and $ \la [U, V] \ra_n=0$ in the case of odd $n$ for the commutator $[U, V]:=UV-VU$ of two arbitrary elements, for example, in \cite{pseudo}. We get the invariance under cyclic permutations as a consequence of the previous properties. The similarity invariance is a consequence of the invariance under cyclic permutations.
\end{proof}
The operation $\la U\ra_0$ is also called the scalar part of $U\in\cl_{p,q}$. This operation is related to the trace of matrices (see Section \ref{sec:3}, and note that the operation $\la \quad \ra_0$ has the same properties (\ref{po}), (\ref{larao}), (\ref{lara02}) as the trace of matrices).

\begin{definition}
We call any operation of the form
\begin{eqnarray}
U \to \sum_{k=0}^n \lambda_k \la U \ra_k,\qquad \lambda_k=\pm 1,\qquad U=\sum_{k=0}^n \la U \ra_k,\qquad \la U \ra_k\in\cl^k_{p,q},\label{opconj}
\end{eqnarray}
an operation of conjugation in Clifford algebra.
\end{definition}
The operations of conjugation commute with each other by definition. The operation of conjugation is an involution: the square of each operation equals the identical operation $\id$ (which is also an operation of conjugation with all $\lambda_k$ equal to $1$). In the theory of Clifford algebras, there are three classical operations of conjugation: the grade involution (or the main involution) $\,\widehat\over\,$, the reversion $\,\widetilde\over\,$, and the superposition of these two operations $\,\widehat{\widetilde\over}\,$, which is called the Clifford conjugation:
\begin{eqnarray}
&&\widehat{U}=\sum_{k=0}^n (-1)^k \la U \ra_k,\qquad \widetilde{U}=\sum_{k=0}^n(-1)^{\frac{k(k-1)}{2}} \la U\ra_k,\qquad \widetilde{\widehat{U}}=\sum_{k=0}^n(-1)^{\frac{k(k+1)}{2}} \la U\ra_k,\label{coc}\\
&&\widehat{UV}=\widehat{U}\widehat{V},\qquad \widetilde{UV}=\widetilde{V}\widetilde{U},\qquad \widehat{\widetilde{UV}}=\widehat{\widetilde{V}} \widehat{\widetilde{U}},\qquad \forall U, V\in\cl_{p,q}.\label{grrev}
\end{eqnarray}
We do not use separate notation for the Clifford conjugation in this paper and write the combination of the two symbols $\,\widehat\over\,$ and $\,\widetilde\over\,$. The reversion and the Clifford conjugation are anti-involutions because they satisfy~(\ref{grrev}).

\begin{definition}\label{def2}
Let us consider the following four subspaces of quaternion types $r=0, 1, 2, 3$ in $\cl_{p,q}$ (see \cite{quat}), which are defined using the grade involution and the reversion\footnote{For example, the subspace $\cl^{\overline{0}}_{p,q}$ consists of elements that are not changing under the grade involution $\widehat{U}=U$ and the reversion $\widetilde{U}=U$, i.e. elements of grades $0$, $4$, $8$, etc. The subspace $\cl^{\overline{1}}_{p,q}$ consists of elements that satisfy $\widehat{U}=-U$ and $\widetilde{U}=U$, i.e. elements of grades $1$, $5$, $9$, etc. Similarly for the other two subspaces. Other properties of these four subspaces are discussed in detail in \cite{quat1, quat2}.}:
\begin{equation}\label{quat}
\cl^{\overline{r}}_{p,q}:=\{U\in\cl_{p,q}: \quad \widehat{U}=(-1)^r U,\quad \widetilde{U}=(-1)^{\frac{r(r-1)}{2}}U\}=\bigoplus_{k=r\mod 4}\cl^k_{p,q},\qquad r=0, 1, 2, 3.
\end{equation}
\end{definition}

We denote the projection of an arbitrary element $U\in\cl_{p,q}$ onto the subspace of quaternion type $k$ by $ \la U \ra_{\overline{k}}$, $k=0, 1, 2, 3$. By Definition \ref{def2}, we get
\begin{eqnarray}
  &&U = \la U \ra_{\overline{0}}+ \la U \ra_{\overline{1}}+ \la U \ra_{\overline{2}}+ \la U \ra_{\overline{3}},\qquad
\widehat{U} = \la U \ra_{\overline{0}}- \la U \ra_{\overline{1}}+ \la U \ra_{\overline{2}}- \la U \ra_{\overline{3}},\label{syst}\\
&&\widetilde{U} = \la U \ra_{\overline{0}}+ \la U \ra_{\overline{1}}- \la U \ra_{\overline{2}}- \la U \ra_{\overline{3}},\qquad
\widehat{\widetilde{U}}= \la U \ra_{\overline{0}}- \la U \ra_{\overline{1}}- \la U \ra_{\overline{2}}+ \la U \ra_{\overline{3}}.\nonumber
\end{eqnarray}
Solving this system of four linear equations, we get\footnote{As one of the anonymous reviewers correctly noted, these formulas can also be used as definitions of projection operations onto the subspaces of quaternion types $0$, $1$, $2$, and $3$.}
\begin{eqnarray*}
  && \la U \ra_{\overline{0}} = \frac{1}{4}(U+\widehat{U}+\widetilde{U}+\widehat{\widetilde{U}}), \qquad
  \la U \ra_{\overline{1}} = \frac{1}{4}(U-\widehat{U}+\widetilde{U}-\widehat{\widetilde{U}}),\\
  &&\la U \ra_{\overline{2}} = \frac{1}{4}(U+\widehat{U}-\widetilde{U}-\widehat{\widetilde{U}}),\qquad
  \la U \ra_{\overline{3}} = \frac{1}{4}(U-\widehat{U}-\widetilde{U}+\widehat{\widetilde{U}}).
\end{eqnarray*}
This means that the projection operations onto the subspaces of quaternion types $k=0, 1, 2, 3$ are determined by the operations $\id$, $\,\widetilde\over\,$, $\,\widehat\over\,$, and $\,\widehat{\widetilde\over}\,$. In $\cl_{p,q}$ with $n=p+q\leq 3$, we can similarly realize projection operations onto the subspaces of fixed grades $\la U \ra_k=\la U \ra_{\overline{k}}$, $k=0, 1, 2, 3$, using only the operations (\ref{coc}) because the concepts of grades and quaternion types are the same in these cases: $\cl^{\overline{k}}_{p,q}=\cl^k_{p,q}$, $k=0, 1, 2, 3$.\footnote{As a consequence, the expressions for the determinant, other characteristic polynomial coefficients, and inverse can be realized using only the three classical operations of conjugations (\ref{coc}) in the cases $n\leq 3$ (see Sections \ref{sec:3} and \ref{sec:4}).} If we want to realize projection operations onto fixed grades $k=0, 1, \ldots, n$ in the cases $n\geq 4$, we need more operations of conjugation. For example, in the case $n=4$, we have $\cl^{\overline{0}}_{p,q}=\cl^0_{p,q}\oplus\cl^4_{p,q}$, $\la U\ra_{\overline{0}}=\la U \ra_0+\la U \ra_4$, and we can not separately realize operations $\la U\ra_0$ and $\la U \ra_4$ using only the operations (\ref{coc}).

\begin{definition}
Let us consider the following\footnote{The equivalence of these two definitions follows from the following fact: the binomial coefficient $C^i_k$ is odd if and only if there are no $1$ in the binary notation of the number $i$ in the digits, in which there is $0$ in the binary notation of the number $k$.} operations of conjugation of special type $\va_1$, $\va_2$, \ldots, $\va_m$, $m:=[\log_2 n]+1$:
\begin{eqnarray}
U^{\va_j}&=&\sum_{k=0}^n (-1)^{C^{2^{j-1}}_k} \la U \ra_k=\sum_{k=0,\, \ldots,\, 2^{j-1}-1 \mod 2^j} \la U \ra_k - \sum_{k=2^{j-1},\, \ldots,\, 2^j-1\mod 2^j} \la U \ra_k,\label{ocst}
\end{eqnarray}
where $C^i_k:={k\choose{i}}=\frac{k!}{i!(k-i)!}$ is the binomial coefficient (for $i>k$, we have $C^i_k=0$ by definition) and $[\log_2 n]$ is the integer part of $\log_2 n$.
\end{definition}

In the particular cases, we get
\begin{eqnarray}
U^{\va_1}&=&\sum_{k=0}^n (-1)^{C^1_k} \la U \ra_k=\sum_{k=0}^n (-1)^k \la U \ra_k=\widehat{U},\qquad n\geq 1;\\
U^{\va_2}&=&\sum_{k=0}^n (-1)^{C^2_k} \la U \ra_k=\sum_{k=0}^n (-1)^\frac{k(k-1)}{2} \la U \ra_k=\widetilde{U},\qquad n\geq 2;\\
U^{\va_3}&=&\sum_{k=0}^n (-1)^{C^4_k} \la U \ra_k=\sum_{k=0, 1, 2, 3\mod 8} \la U \ra_k -\sum_{k=4, 5, 6, 7\mod 8} \la U \ra_k ,\qquad n\geq 4;\\
U^{\va_4} &=&\sum_{k=0}^n (-1)^{C^8_k} \la U \ra_k=\sum_{k=0, 1, \ldots, 7\mod 16} \la U \ra_k -\sum_{k=8, 9, \ldots, 15\mod 16}\la U \ra_k,\qquad n\geq 8.
\end{eqnarray}
We see that the first two operations coincide with the two classical operations -- the grade involution $\va_1=\widehat\over\,$ and the reversion $\va_2=\widetilde\over\,$. We denote the superposition of the operations $\va_k$ and $\va_l$ by $\va_k \va_l$. The definitions of the operations (\ref{ocst}) and their superpositions are illustrated in Table \ref{table1} (we put ``$+$'' if $\la U \ra_k \to \la U \ra_k$ under the corresponding operation and put ``$-$'' if $\la U \ra_k \to -\la U \ra_k$ under the corresponding operation for each grade $k=0, 1, \ldots, n$). In the case $n=1$, we can realize projection operations onto the subspaces of fixed grades $0$ and $1$ using only the identity operation and the grade involution (because the $2\times 2$ matrix in the upper left corner of Table \ref{table1} is invertible; we interpret ``$+$'' as $1$ and ``$-$'' as $-1$ ). In the cases $n=2, 3$, we need also the operation $\,\widetilde\over\,$ (see the invertible $4\times 4$  matrix in the upper left corner of  Table \ref{table1}; this matrix corresponds to the system of equations (\ref{syst})). In the cases $n=4, 5, 6, 7$, we can do this using the first three operations $\,\widehat\over\,$, $\,\widetilde\over\,$, $\va_3$, and their superpositions (see the invertible $8\times 8$ matrix in the upper left corner of  Table \ref{table1}). In the cases $n=8, \ldots, 15$, we need also the fourth operation $\va_4$ (see the invertible $16\times 16$ matrix, which corresponds to the whole Table \ref{table1}), and so on. As we will see below, explicit formulas for the determinant, other characteristic polynomial coefficients, and inverse can be written using only the presented here operations of conjugation. We use the notation $\va:=\va_3$ further.

\begin{table}
\begin{center}
\begin{tabular}{|c|l|l|l|l|l|l|l|l|l|l|l|l|l|l|l|l|}
  \hline
grade $k$ & 0 & 1 & 2 & 3 & 4 & 5 & 6 & 7 & 8 & 9 & 10 & 11 & 12 & 13 & 14 & 15 \\ \hline
\multicolumn{1}{|c}{$\id$} &
\multicolumn{1}{|c}{$+$} & \multicolumn{1}{c}{$+$} & \multicolumn{1}{|c}{$+$} & \multicolumn{1}{c}{$+$} & \multicolumn{1}{|c}{$+$} & \multicolumn{1}{c}{$+$} & \multicolumn{1}{c}{$+$} & \multicolumn{1}{c}{$+$} & \multicolumn{1}{|c}{$+$} & \multicolumn{1}{c}{$+$} & \multicolumn{1}{c}{$+$} & \multicolumn{1}{c}{$+$} & \multicolumn{1}{c}{$+$} & \multicolumn{1}{c}{$+$} & \multicolumn{1}{c}{$+$} & \multicolumn{1}{c|}{$+$} \\ \cline{1-1}
\multicolumn{1}{|c}{${\va_1}=\widehat\over\,$} &
\multicolumn{1}{|c}{$+$} & \multicolumn{1}{c}{$-$} & \multicolumn{1}{|c}{$+$} & \multicolumn{1}{c}{$-$} & \multicolumn{1}{|c}{$+$} & \multicolumn{1}{c}{$-$} & \multicolumn{1}{c}{$+$} & \multicolumn{1}{c}{$-$} & \multicolumn{1}{|c}{$+$} & \multicolumn{1}{c}{$-$} & \multicolumn{1}{c}{$+$} & \multicolumn{1}{c}{$-$} & \multicolumn{1}{c}{$+$} & \multicolumn{1}{c}{$-$} & \multicolumn{1}{c}{$+$} & \multicolumn{1}{c|}{$-$}  \\ \cline{1-3}
\multicolumn{1}{|c}{${\va_2}=\widetilde\over\,$} &
\multicolumn{1}{|c}{$+$} & \multicolumn{1}{c}{$+$} & \multicolumn{1}{c}{$-$} & \multicolumn{1}{c}{$-$} & \multicolumn{1}{|c}{$+$} & \multicolumn{1}{c}{$+$} & \multicolumn{1}{c}{$-$} & \multicolumn{1}{c}{$-$} & \multicolumn{1}{|c}{$+$} & \multicolumn{1}{c}{$+$} & \multicolumn{1}{c}{$-$} & \multicolumn{1}{c}{$-$} & \multicolumn{1}{c}{$+$} & \multicolumn{1}{c}{$+$} & \multicolumn{1}{c}{$-$} & \multicolumn{1}{c|}{$-$} \\ \cline{1-1}
\multicolumn{1}{|c}{${\va_1 \va_2}$} &
\multicolumn{1}{|c}{$+$} & \multicolumn{1}{c}{$-$} & \multicolumn{1}{c}{$-$} & \multicolumn{1}{c}{$+$} & \multicolumn{1}{|c}{$+$} & \multicolumn{1}{c}{$-$} & \multicolumn{1}{c}{$-$} & \multicolumn{1}{c}{$+$} & \multicolumn{1}{|c}{$+$} & \multicolumn{1}{c}{$-$} & \multicolumn{1}{c}{$-$} & \multicolumn{1}{c}{$+$} & \multicolumn{1}{c}{$+$} & \multicolumn{1}{c}{$-$} & \multicolumn{1}{c}{$-$} & \multicolumn{1}{c|}{$+$} \\ \cline{1-5}
\multicolumn{1}{|c}{$\va_3$} &
\multicolumn{1}{|c}{$+$} & \multicolumn{1}{c}{$+$} & \multicolumn{1}{c}{$+$} & \multicolumn{1}{c}{$+$} & \multicolumn{1}{c}{$-$} & \multicolumn{1}{c}{$-$} & \multicolumn{1}{c}{$-$} & \multicolumn{1}{c}{$-$} & \multicolumn{1}{|c}{$+$} & \multicolumn{1}{c}{$+$} & \multicolumn{1}{c}{$+$} & \multicolumn{1}{c}{$+$} & \multicolumn{1}{c}{$-$} & \multicolumn{1}{c}{$-$} & \multicolumn{1}{c}{$-$} & \multicolumn{1}{c|}{$-$} \\ \cline{1-1}
\multicolumn{1}{|c}{${\va_1 \va_3}$} &
\multicolumn{1}{|c}{$+$} & \multicolumn{1}{c}{$-$} & \multicolumn{1}{c}{$+$} & \multicolumn{1}{c}{$-$} & \multicolumn{1}{c}{$-$} & \multicolumn{1}{c}{$+$} & \multicolumn{1}{c}{$-$} & \multicolumn{1}{c}{$+$} & \multicolumn{1}{|c}{$+$} & \multicolumn{1}{c}{$-$} & \multicolumn{1}{c}{$+$} & \multicolumn{1}{c}{$-$} & \multicolumn{1}{c}{$-$} & \multicolumn{1}{c}{$+$} & \multicolumn{1}{c}{$-$} & \multicolumn{1}{c|}{$+$} \\ \cline{1-1}
\multicolumn{1}{|c}{${\va_2 \va_3}$} &
\multicolumn{1}{|c}{$+$} & \multicolumn{1}{c}{$+$} & \multicolumn{1}{c}{$-$} & \multicolumn{1}{c}{$-$} & \multicolumn{1}{c}{$-$} & \multicolumn{1}{c}{$-$} & \multicolumn{1}{c}{$+$} & \multicolumn{1}{c}{$+$} & \multicolumn{1}{|c}{$+$} & \multicolumn{1}{c}{$+$} & \multicolumn{1}{c}{$-$} & \multicolumn{1}{c}{$-$} & \multicolumn{1}{c}{$-$} & \multicolumn{1}{c}{$-$} & \multicolumn{1}{c}{$+$} & \multicolumn{1}{c|}{$+$} \\ \cline{1-1}
\multicolumn{1}{|c}{$\va_1 \va_2 \va_3$} &
\multicolumn{1}{|c}{$+$} & \multicolumn{1}{c}{$-$} & \multicolumn{1}{c}{$-$} & \multicolumn{1}{c}{$+$} & \multicolumn{1}{c}{$-$} & \multicolumn{1}{c}{$+$} & \multicolumn{1}{c}{$+$} & \multicolumn{1}{c}{$-$} & \multicolumn{1}{|c}{$+$} & \multicolumn{1}{c}{$-$} & \multicolumn{1}{c}{$-$} & \multicolumn{1}{c}{$+$} & \multicolumn{1}{c}{$-$} & \multicolumn{1}{c}{$+$} & \multicolumn{1}{c}{$+$} & \multicolumn{1}{c|}{$-$} \\ \cline{1-9}
\multicolumn{1}{|c}{$\va_4$} &
\multicolumn{1}{|c}{$+$} & \multicolumn{1}{c}{$+$} & \multicolumn{1}{c}{$+$} & \multicolumn{1}{c}{$+$} & \multicolumn{1}{c}{$+$} & \multicolumn{1}{c}{$+$} & \multicolumn{1}{c}{$+$} & \multicolumn{1}{c}{$+$} & \multicolumn{1}{c}{$-$} & \multicolumn{1}{c}{$-$} & \multicolumn{1}{c}{$-$} & \multicolumn{1}{c}{$-$} & \multicolumn{1}{c}{$-$} & \multicolumn{1}{c}{$-$} & \multicolumn{1}{c}{$-$} & \multicolumn{1}{c|}{$-$} \\ \cline{1-1}
\multicolumn{1}{|c}{${\va_1 \va_4}$} &
\multicolumn{1}{|c}{$+$} & \multicolumn{1}{c}{$-$} & \multicolumn{1}{c}{$+$} & \multicolumn{1}{c}{$-$} & \multicolumn{1}{c}{$+$} & \multicolumn{1}{c}{$-$} & \multicolumn{1}{c}{$+$} & \multicolumn{1}{c}{$-$} & \multicolumn{1}{c}{$-$} & \multicolumn{1}{c}{$+$} & \multicolumn{1}{c}{$-$} & \multicolumn{1}{c}{$+$} & \multicolumn{1}{c}{$-$} & \multicolumn{1}{c}{$+$} & \multicolumn{1}{c}{$-$} & \multicolumn{1}{c|}{$+$} \\ \cline{1-1}
\multicolumn{1}{|c}{${\va_2 \va_4}$} &
\multicolumn{1}{|c}{$+$} & \multicolumn{1}{c}{$+$} & \multicolumn{1}{c}{$-$} & \multicolumn{1}{c}{$-$} & \multicolumn{1}{c}{$+$} & \multicolumn{1}{c}{$+$} & \multicolumn{1}{c}{$-$} & \multicolumn{1}{c}{$-$} & \multicolumn{1}{c}{$-$} & \multicolumn{1}{c}{$-$} & \multicolumn{1}{c}{$+$} & \multicolumn{1}{c}{$+$} & \multicolumn{1}{c}{$-$} & \multicolumn{1}{c}{$-$} & \multicolumn{1}{c}{$+$} & \multicolumn{1}{c|}{$+$} \\ \cline{1-1}
\multicolumn{1}{|c}{${\va_1 \va_2 \va_4}$} &
\multicolumn{1}{|c}{$+$} & \multicolumn{1}{c}{$-$} & \multicolumn{1}{c}{$-$} & \multicolumn{1}{c}{$+$} & \multicolumn{1}{c}{$+$} & \multicolumn{1}{c}{$-$} & \multicolumn{1}{c}{$-$} & \multicolumn{1}{c}{$+$} & \multicolumn{1}{c}{$-$} & \multicolumn{1}{c}{$+$} & \multicolumn{1}{c}{$+$} & \multicolumn{1}{c}{$-$} & \multicolumn{1}{c}{$-$} & \multicolumn{1}{c}{$+$} & \multicolumn{1}{c}{$+$} & \multicolumn{1}{c|}{$-$} \\ \cline{1-1}
\multicolumn{1}{|c}{$\va_3 \va_4$} &
\multicolumn{1}{|c}{$+$} & \multicolumn{1}{c}{$+$} & \multicolumn{1}{c}{$+$} & \multicolumn{1}{c}{$+$} & \multicolumn{1}{c}{$-$} & \multicolumn{1}{c}{$-$} & \multicolumn{1}{c}{$-$} & \multicolumn{1}{c}{$-$} & \multicolumn{1}{c}{$-$} & \multicolumn{1}{c}{$-$} & \multicolumn{1}{c}{$-$} & \multicolumn{1}{c}{$-$} & \multicolumn{1}{c}{$+$} & \multicolumn{1}{c}{$+$} & \multicolumn{1}{c}{$+$} & \multicolumn{1}{c|}{$+$} \\ \cline{1-1}
\multicolumn{1}{|c}{${\va_1 \va_3 \va_4}$} &
\multicolumn{1}{|c}{$+$} & \multicolumn{1}{c}{$-$} & \multicolumn{1}{c}{$+$} & \multicolumn{1}{c}{$-$} & \multicolumn{1}{c}{$-$} & \multicolumn{1}{c}{$+$} &\multicolumn{1}{c}{ $-$} & \multicolumn{1}{c}{$+$} & \multicolumn{1}{c}{$-$} & \multicolumn{1}{c}{$+$} & \multicolumn{1}{c}{$-$} & \multicolumn{1}{c}{$+$} & \multicolumn{1}{c}{$+$} & \multicolumn{1}{c}{$-$} & \multicolumn{1}{c}{$+$} &\multicolumn{1}{c|}{$-$}  \\ \cline{1-1}
\multicolumn{1}{|c}{${\va_2 \va_3 \va_4}$} &
\multicolumn{1}{|c}{$+$} & \multicolumn{1}{c}{$+$} & \multicolumn{1}{c}{$-$} & \multicolumn{1}{c}{$-$} & \multicolumn{1}{c}{$-$} & \multicolumn{1}{c}{$-$} & \multicolumn{1}{c}{$+$} & \multicolumn{1}{c}{$+$} & \multicolumn{1}{c}{$-$} & \multicolumn{1}{c}{$-$} & \multicolumn{1}{c}{$+$} & \multicolumn{1}{c}{$+$} & \multicolumn{1}{c}{$+$} & \multicolumn{1}{c}{$+$} & \multicolumn{1}{c}{$-$} & \multicolumn{1}{c|}{$-$}  \\ \cline{1-1}
\multicolumn{1}{|c}{$\va_1 \va_2 \va_3 \va_4$} &
\multicolumn{1}{|c}{$+$} & \multicolumn{1}{c}{$-$} & \multicolumn{1}{c}{$-$} & \multicolumn{1}{c}{$+$} & \multicolumn{1}{c}{$-$} & \multicolumn{1}{c}{$+$} & \multicolumn{1}{c}{$+$} & \multicolumn{1}{c}{$-$} & \multicolumn{1}{c}{$-$} & \multicolumn{1}{c}{$+$} & \multicolumn{1}{c}{$+$} & \multicolumn{1}{c}{$-$} & \multicolumn{1}{c}{$+$} & \multicolumn{1}{c}{$-$} & \multicolumn{1}{c}{$-$} & \multicolumn{1}{c|}{$+$} \\ \cline{1-17}
\end{tabular}
\end{center}\caption{The identity operation $\id$ and the operation $\va_1$ for $n\leq 1$; the operations $\id$, $\va_1$, $\va_2$, $\va_1 \va_2$ for $n\leq 3$; the operations $\va_1$, $\va_2$, $\va_3$ and their superpositions for $n\leq 7$; the operations $\va_1$, $\va_2$, $\va_3$, $\va_4$ and their superpositions for $n\leq 15$. The table can be continued.}\label{table1}
\end{table}

In the following theorem, we give explicit formulas for the operation $\la \quad \ra_0$ using only the operations $\va_1$, \ldots, $\va_m$ and their superpositions. Different explicit formulas for the projection onto the subspace of grade~$0$ correspond to different explicit formulas for the trace, determinant, and other characteristic polynomial coefficients of the Clifford algebra elements (see Section \ref{sec:4}).

\begin{theorem}\label{th10} We can realize the operation $\la\quad\ra_0$ using the operations $\va_1$, $\va_2$, \ldots $\va_m$ in the following form:
\begin{eqnarray}
\la U \ra_0=\frac{1}{2^m}(U+U^{\va_1}+U^{\va_2}+\cdots+U^{\va_1 \ldots \va_m}),\qquad m=[\log_2 n]+1,\label{lara0}
\end{eqnarray}
in particular,
\begin{eqnarray}
\la U \ra_0&=&\frac{1}{2}(U+\widehat{U}),\qquad n=1;\label{r1}\\
\la U \ra_0&=&\frac{1}{4}(U+\widehat{U}+\widetilde{U}+\widehat{\widetilde{U}}),\qquad n=2, 3;\label{r3}\\
\la U \ra_0 &=&\frac{1}{8}(U+\widehat{U}+\widetilde{U}+\widehat{\widetilde{U}}+U^\va+\widehat{U}^\va+\widetilde{U}^\va+\widehat{\widetilde{U}}^\va), \qquad n=4, 5, 6, 7.
\end{eqnarray}
In some cases, the operation $\la\quad \ra_0$ can be realized in the following simpler form:
\begin{eqnarray}
\la U \ra_0&=&\frac{1}{2}(U+\widehat{\widetilde{U}})=\frac{1}{2}(\widehat{U}+\widetilde{U}),\qquad n=2;\label{r2}\\
\la U \ra_0&=&\frac{1}{4}(U+\widehat{\widetilde{U}}+\widehat{U}^\va+\widetilde{U}^\va)=\frac{1}{4}(\widehat{U}+\widetilde{U}+U^\va+ \widehat{\widetilde{U}}^\va),\qquad n=4, 5, 6;\label{T1}\\
\la U \ra_0&=&\frac{1}{4}(U+\widehat{U}+\widetilde{U}^\va+\widehat{\widetilde{U}}^\va)= \frac{1}{4}(\widetilde{U}+\widehat{\widetilde{U}}+U^\va+\widehat{U}^\va),\qquad n=4, 5;\\
\la U \ra_0&=&\frac{1}{4}(U+\widetilde{U}+\widehat{U}^\va+\widehat{\widetilde{U}}^\va)=\frac{1}{4}(\widehat{U}+\widehat{\widetilde{U}}+U^\va+ \widetilde{U}^\va),\qquad n=4.\label{T2}
\end{eqnarray}
The same expressions coincide with the projection onto the center $\cen(\cl_{p,q})$ in the cases:
\begin{eqnarray}
\la U \ra_{\cen}&=&\la U\ra_0+\la U\ra_n= \frac{1}{2}(U+\widehat{\widetilde{U}}),\qquad n=3;\\
\la U \ra_{\cen}&=& \la U\ra_0+\la U\ra_n= \frac{1}{4}(U+\widehat{\widetilde{U}}+\widehat{U}^\va+\widetilde{U}^\va),\qquad n=7;\label{T3}\\
\la U \ra_{\cen}&=& \la U\ra_0+\la U\ra_n= \frac{1}{4}(U+\widetilde{U}+\widehat{U}^\va+\widehat{\widetilde{U}}^\va),\qquad n=5.\label{T4}
\end{eqnarray}
\end{theorem}
\begin{proof} We get the formula (\ref{lara0}) using the following fact. We have the same number of pluses and minuses in each column, except the first one, of each of the considered square matrices (of dimension 2, 4, 8, 16, \ldots, $2^{[\log_2 n]+1}$) in the upper left corner of Table \ref{table1}. This can be proved by induction: this is true for the first matrix of dimension 2; each of the considered square matrices is equal to the block-diagonal matrix $$\left(\begin{array}{cc} A & A \\ A & -A \\ \end{array} \right),$$ where $A$ is the previous square matrix (by the definition of the operations $\va_1$, \ldots, $\va_m$). We get all other formulas for $\la\quad\ra_0$ and $\la\quad\ra_{\cen}$ using the definitions of the operations $\va_j$, $j=1, \ldots, m$ in the particular cases $n\leq 7$.
\end{proof}

Note that for fixed $n$, there are $2^{n+1}$ different operations of conjugation (\ref{opconj}). The grade-negation operations $$U_{\underline{k}}:=U-2\la U \ra_k,\qquad k=0, 1, \ldots, n,$$
which are used in \cite{dadbeh,acus,hitzer1,hitzer2}, are the particular cases of the operations of conjugation (\ref{opconj}).  We can consider $n+1$ grade-negation operations and realize the other operations of conjugation as superpositions of these operations. Alternatively, we can consider $m=[\log_2 n]+1$ (which is less than $n+1$) operations of conjugation of special type $\va_1$, \ldots, $\va_m$ and realize the other operations of conjugation as linear combinations of superpositions of these operations (for example, we have different realization of the operation of conjugation (\ref{over}) using the operations $\va_1, \ldots, \va_m$, see Lemma \ref{lemover1})\footnote{It can be proved that all operations $\la U \ra_k$, $k=1, \ldots, n$ (similarly to the case of the operation $\la U\ra_0$, see Theorem \ref{th10}) can be realized as linear combinations of the operations $\id$, $\va_1$, \ldots, $\va_m$, $\va_1 \va_2$, \ldots, $\va_1\cdots \va_m$. As a consequence, we get that all operations of conjugation (\ref{opconj}) can be realized as linear combinations of the operations $\id$, $\va_1$, \ldots, $\va_m$, $\va_1 \va_2$, \ldots, $\va_1\cdots \va_m$. We do not use this fact in this paper.}.

Note that for the operations $\va_j$, $j=3, 4, \ldots$ (we call them additional operations of conjugation), we have $(UV)^{\va_j}\neq U^{\va_j} V^{\va_j}$ and $(UV)^{\va_j}\neq V^{\va_j} U^{\va_j} $ in the general case. Let us present the following nontrivial properties of the operation $\va:=\va_3$.  We use these properties in Sections \ref{sec:3} and \ref{sec:4} of this paper.

\begin{theorem} We have
\begin{eqnarray}
U (\widetilde{U} \widehat{U})^\va=(\widehat{U} \widetilde{U})^\va U,\quad
\widehat{U} (\widehat{\widetilde{U}} U)^\va=(U \widehat{\widetilde{U}})^\va \widehat{U},\quad
\widetilde{U}(U \widehat{\widetilde{U}})^\va=(\widehat{\widetilde{U}} U)^\va \widetilde{U},\quad
\widehat{\widetilde{U}} (\widehat{U}\widetilde{U})^\va=(\widetilde{U}\widehat{U})^\va\widehat{\widetilde{U}},\quad n\leq 7;\label{T5}\\
U (\widehat{\widetilde{U}} \widehat{U})^\va=(\widehat{U}\widehat{\widetilde{U}})^\va U,\quad
\widehat{U}(\widetilde{U} U)^\va=(U \widetilde{U})^\va \widehat{U},\quad
\widetilde{U}(\widehat{U} \widehat{\widetilde{U}})^\va=(\widehat{\widetilde{U}} \widehat{U})^\va \widetilde{U},\quad
\widehat{\widetilde{U}} (U\widetilde{U})^\va=(\widetilde{U} U)^\va\widehat{\widetilde{U}}, \quad n\leq 5.\label{T6}
\end{eqnarray}
\end{theorem}
\begin{proof}
In the cases $n\leq 6$, using (\ref{T1}), (\ref{lara0}) and substituting $V=\widehat{\widetilde{U}}$, we get
\begin{eqnarray*}
&&\frac{1}{4}(UV+\widehat{\widetilde{UV}}+(\widehat{UV})^\va+(\widetilde{UV})^\va)=\la UV \ra_0=\la VU \ra_0= \frac{1}{4}(VU+\widehat{\widetilde{VU}}+(\widehat{VU})^\va+(\widetilde{VU})^\va)\in\cl^0_{p,q},\\
&&U\widehat{\widetilde{U}}+U\widehat{\widetilde{U}}+(\widehat{U}\widetilde{U})^\va+(\widehat{U}\widetilde{U})^\va= \widehat{\widetilde{U}}U+\widehat{\widetilde{U}}U+(\widetilde{U}\widehat{U})^\va+(\widetilde{U}\widehat{U})^\va\in\cl^0_{p,q}.
\end{eqnarray*}
The expressions on the left side and on the right side are scalars. We can multiply the left side by $U$ on the right and the right side by $U$ on the left and get
$(\widehat{U}\widetilde{U})^\va U=U(\widetilde{U}\widehat{U})^\va$. Taking the grade involution, the reversion, or superposition of these two operations, we obtain the other formulas (\ref{T5}). We get the same in the case $n=7$ using (\ref{T3}) and the property $\la UV \ra_{\cen}=\la VU \ra_{\cen}$.

In the cases $n\leq 4$, using (\ref{T2}), (\ref{lara0}) and substituting $V=\widetilde{U}$, we get
\begin{eqnarray*}
&&\frac{1}{4}(UV+\widetilde{UV}+(\widehat{UV})^\va+(\widehat{\widetilde{UV}})^\va)=\la UV \ra_0=\la VU \ra_0= \frac{1}{4}(VU+\widetilde{VU}+(\widehat{VU})^\va+(\widehat{\widetilde{VU}})^\va)\in\cl^0_{p,q},\\
&&U\widetilde{U}+U\widetilde{U}+(\widehat{U}\widehat{\widetilde{U}})^\va+(\widehat{U}\widehat{\widetilde{U}})^\va= \widetilde{U}U+\widetilde{U}U+(\widehat{\widetilde{U}}\widehat{U})^\va+(\widehat{\widetilde{U}}\widehat{U})^\va\in\cl^0_{p,q}.
\end{eqnarray*}
Multiplying the left side by $U$ on the right and the left side by $U$ on the left, we get
$(\widehat{U}\widehat{\widetilde{U}})^\va U=U (\widehat{\widetilde{U}}\widehat{U})^\va$. Taking the grade involution, the reversion, or superposition of these two operations, we obtain the other formulas (\ref{T6}). We get the same in the case $n=5$ using (\ref{T4}) and the property $\la UV \ra_{\cen}=\la VU \ra_{\cen}$.
\end{proof}

\begin{definition}
Let us consider one other operation of conjugation that will be useful for the purposes of this paper:
\begin{eqnarray}
\overline{U}:=\la U \ra_0-\sum_{k=1}^n  \la U \ra_k,\qquad U\in\cl_{p,q}.\label{over}
\end{eqnarray}
\end{definition}

We denote this operation by $\stackrel{\over\quad}{\quad}$ because this operation is an analogue of the complex conjugation in the case $\cl_{0,1}\cong\C$ (and coincides with the grade involution $\overline{U}=\widehat{U}$) and is an analogue of the quaternion conjugation in the case $\cl_{0,2}\cong\H$ (and coincides with the Clifford conjugation $\overline{U}=\widehat{\widetilde{U}}$)\footnote{Note that some authors \cite{Lounesto} denote by $\stackrel{\over\quad}{\quad}$ the operation of Clifford conjugation. We denote the Clifford conjugation by two symbols $\,\widetilde{\widehat\over}\,$ in this paper so that there is no confusion.}.

\begin{lemma}\label{lemover1} We can realize the operation $\stackrel{\over\quad}{\quad}$ using the operations $\va_1$, $\va_2$, \ldots $\va_m$ in the following form:
$$\overline{U}=\frac{1}{2^{m-1}}((1-2^{m-1})U+U^{\va_1}+U^{\va_2}+\cdots+U^{\va_1 \ldots \va_m}),\qquad
m=[\log_2 n]+1,$$
in particular,
\begin{eqnarray}
\overline{U}&=&\widehat{U},\qquad n=1;\\
\overline{U}&=&\frac{1}{2}(\widehat{U}+\widetilde{U}+\widehat{\widetilde{U}}-U),\qquad n=2, 3;\\
\overline{U} &=&\frac{1}{4}(\widehat{U}+\widetilde{U}+\widehat{\widetilde{U}}+U^\va+\widehat{U}^\va+\widetilde{U}^\va+\widehat{\widetilde{U}}^\va-3U), \qquad n=4, 5, 6, 7.
\end{eqnarray}
In some cases, the operation $\stackrel{\over\quad}{\quad}$ can be realized in the following simpler form:
\begin{eqnarray}
\overline{U}&=&\widetilde{\widehat{U}},\qquad n=2;\\
\overline{U}&=&\frac{1}{2}(\widehat{U}^{\va}+ \widetilde{U}^\va+\widehat{\widetilde{U}}-U),\qquad n=4, 5, 6;\label{over6}\\
\overline{U}&=&\frac{1}{2}(\widehat{U}+\widetilde{U}^\va+\widehat{\widetilde{U}}^\va-U),\qquad n=4, 5;\\
\overline{U}&=&\frac{1}{2}(\widehat{U}^{\va}+ \widetilde{U}+\widehat{\widetilde{U}}^{\va}-U),\qquad n=4.
\end{eqnarray}
\end{lemma}

\begin{proof} We use $\overline{U}=2 \la U \ra_0-U$ and different realizations of the operation $\la\quad \ra_0$ from Theorem \ref{th10}. For example, using (\ref{lara0}), we get
\begin{eqnarray*}\overline{U}&=&\la U \ra_0-\sum_{k=1}^n \la U \ra_k=2 \la U \ra_0-U=\frac{1}{2^{m-1}} (U+U^{\va_1}+U^{\va_2}+\cdots+U^{\va_1 \ldots \va_m})-U\\
&=&\frac{1}{2^{m-1}}((1-2^{m-1})U+U^{\va_1}+U^{\va_2}+\cdots+U^{\va_1 \ldots \va_m})
\end{eqnarray*}
for $2^{m-1}-1 < n \leq 2^m-1$. We obtain the other formulas analogously.
\end{proof}

The operation $\stackrel{\over\quad}{\quad}$ has the following property in the case of arbitrary $n$.

\begin{lemma}\label{lemmaover}
We have
\begin{eqnarray}
\overline{UV}U=U\overline{VU},\qquad \forall U,V\in\cl_{p,q}.
\end{eqnarray}
In the particular case, $\overline{U}U=U\overline{U}$, $\forall U\in\cl_{p,q}$.
\end{lemma}
\begin{proof}
We have
$$\frac{UV+\overline{UV}}{2}=\la UV \ra_0=\la VU \ra_0=\frac{VU+\overline{VU}}{2}.$$
The expressions on the left side and on the right side are scalars. Multiplying the left side by $2U$ on the right and the right side by $2U$ on the left, we get $\overline{UV}U=U\overline{VU}$. Substituting $V=e$, we get $\overline{U}U=U\overline{U}$.
\end{proof}
\section{Functionals of special form and inverses in $\cl_{p,q}$ with $n=p+q\leq 6$ }
\label{sec:3}

Let us call an arbitrary function $N(U):\cl_{p,q}\to \cl^0_{p,q}\equiv \R$ with values in the subspace of grade~$0$ a functional\footnote{In the literature \cite{dadbeh, acus, hitzer1, LS}, such expressions or special cases of such expressions are also called norms in Clifford algebras, norm functions, determinant norms, scalars, etc.} in Clifford algebra (note that it can be non-linear). We are mainly interested in functionals of the special form $N(U)=U F(U)$, where the non-trivial function $F(U):\cl_{p,q}\to\cl_{p,q}$ contains only the operations of summation, multiplication, and the operations of conjugation (\ref{opconj}). Such functionals give us the explicit formulas for the inverse of the Clifford algebra element $U^{-1}=\frac{F(U)}{N(U)}$, where we identify elements of grade $0$ with scalars $\cl^0_{p,q}\equiv \R$, $e\equiv 1$. Note that $N(U)=U F(U)=F(U) U$ because the left inverse coincides with the right inverse in Clifford algebras. In Section \ref{sec:4}, we show that all functionals $N(U)$ considered in this section coincide with the determinant $\Det(U)$ of the Clifford algebra element $U\in\cl_{p,q}$ (generalization of the concept of the determinant of matrices) and the corresponding functions $F(U)$ coincide with the adjugate $\Adj(U)$ of the Clifford algebra element $U$.

Below we present the explicit expressions for functionals of special form in the cases $n\leq 6$. The formulas for the cases $n=1, 2, 3$ are known. The presented new formulas for the cases $n=4, 5, 6$ use standard operations of conjugation $\,\widehat\over\,$, $\,\widetilde\over\,$ and only one additional operation $\va$ (they do not use grade-negation operations, compare with the known formulas for the cases $n=4, 5, 6$ in the papers \cite{dadbeh, LS, rudn, hitzer1, acus}). We give an analytical proof of all formulas using the properties of the additional operation of conjugation $\va$ (see the previous section) and the method of quaternion typification. We do not use the exterior product, the left and right contractions (see \cite{dadbeh, hitzer1, LS}) in our considerations.

One of the key points of the method of quaternion typification (see \cite{quat, quat1, quat2}) is that the Clifford algebra $\cl_{p,q}$ is a $Z_2\times Z_2$-graded algebra w.r.t. the four subspaces (\ref{quat}) and the operations of commutator $[U,V]=UV-VU$ and anticommutator $\{U, V\}=UV+VU$:
\begin{eqnarray}
&&[\cl^{\overline{k}}_{p,q}, \cl^{\overline{k}}_{p,q}]\subset \cl^{\overline{2}}_{p,q},\quad [\cl^{\overline{k}}_{p,q}, \cl^{\overline{2}}_{p,q}]\subset \cl^{\overline{k}}_{p,q},\quad k=0, 1, 2, 3,\nonumber\\
&&[\cl^{\overline{0}}_{p,q}, \cl^{\overline{1}}_{p,q}]\subset \cl^{\overline{3}}_{p,q},\quad [\cl^{\overline{0}}_{p,q}, \cl^{\overline{3}}_{p,q}]\subset \cl^{\overline{1}}_{p,q},\quad [\cl^{\overline{1}}_{p,q}, \cl^{\overline{3}}_{p,q}]\subset \cl^{\overline{0}}_{p,q};\label{z2z2}\\
&&\{\cl^{\overline{k}}_{p,q}, \cl^{\overline{k}}_{p,q}\}\subset \cl^{\overline{0}}_{p,q},\quad \{\cl^{\overline{k}}_{p,q}, \cl^{\overline{0}}_{p,q}\}\subset \cl^{\overline{k}}_{p,q},\quad k=0, 1, 2, 3,\nonumber\\
&&\{\cl^{\overline{1}}_{p,q}, \cl^{\overline{2}}_{p,q}\}\subset \cl^{\overline{3}}_{p,q},\quad \{\cl^{\overline{2}}_{p,q}, \cl^{\overline{3}}_{p,q}\}\subset \cl^{\overline{1}}_{p,q},\quad \{\cl^{\overline{3}}_{p,q}, \cl^{\overline{1}}_{p,q}\}\subset \cl^{\overline{2}}_{p,q}.\nonumber
\end{eqnarray}
As a particular case, we have $U^2=\frac{1}{2}\{U, U\}\in \cl^{\overline{0}}_{p,q}$ for arbitrary $U\in\cl^{\overline{k}}_{p,q}$, $k=0, 1, 2, 3$. Also we use some other simple facts on grades of different expressions in Clifford algebras, see \cite{pseudo}. For example, the product of two elements of grades $k$ and $l$, $k\geq l$, is the sum of elements of grades $k-l$, $k-l+2$, $k-l+4$, \ldots, $k+l$.

\begin{lemma}\label{HJ} For an arbitrary Clifford algebra element $U\in\cl_{p,q}$, we have
$$U \widetilde{U}\in\cl^{\overline{0}}_{p,q}\oplus\cl^{\overline{1}}_{p,q},\qquad \widetilde{U}U\in\cl^{\overline{0}}_{p,q}\oplus\cl^{\overline{1}}_{p,q},\qquad U \widehat{\widetilde{U}}\in\cl^{\overline{0}}_{p,q}\oplus\cl^{\overline{3}}_{p,q},\qquad \widehat{\widetilde{U}}U\in \cl^{\overline{0}}_{p,q}\oplus\cl^{\overline{3}}_{p,q}.$$
\end{lemma}
\begin{proof} Using (\ref{grrev}), we get
\begin{eqnarray*}
\widetilde{U \widetilde{U}}=\widetilde{\widetilde{U}} \widetilde{U} = U \widetilde{U},\qquad \widetilde{\widetilde{U}U }=\widetilde{U}\widetilde{\widetilde{U}} =\widetilde{U} U ,\qquad
\widehat{\widetilde{U \widehat{\widetilde{U}}}}=\widehat{\widetilde{\widehat{\widetilde{U}}}} \widehat{\widetilde{U}} = U\widehat{\widetilde{U}},\qquad \widehat{\widetilde{\widehat{\widetilde{U}}U}}=\widehat{\widetilde{U}}\widehat{\widetilde{\widehat{\widetilde{U}}}} = \widehat{\widetilde{U}}U.
\end{eqnarray*}
This means that the considered expressions do not change under the reversion or under the Clifford conjugation and belong to the corresponding subspaces of quaternion types by the definition~(\ref{quat}).
\end{proof}

Note that in the case of arbitrary $n$, there exist functionals $N:\cl_{p,q}\to\cl^0_{p,q}$ of the special form $N(U)=U F(U)$, where nontrivial function $F(U)$ contains only the operations of multiplication, summation, and $m=[\log_2 n]+1$ operations of conjugation $\va_1$, $\va_2$, \ldots, $\va_m$. In Theorem \ref{thNF}, we give the explicit form of these functionals for the cases $n\leq 5$ (and the operation of summation is not needed in these cases). In Lemma \ref{lemn6}, we give the explicit form of $N(U)$ for the case $n=6$. For the cases $n\geq 7$, the existence of such functionals (which equal $\Det(U)$) follows from the results of Section \ref{sec:4}. The method of construction of such functionals (and explicit formulas) in the case of arbitrary $n$ is also given in Section \ref{sec:4}.

Let us use the following notation for the expressions $H:=U\widetilde{U}$ and $J:=U\widehat{\widetilde{U}}$. We omit the brackets $U_k:=\la U \ra_k \in\cl^k_{p,q}$ to simplify notations for the projection operations onto subspaces of fixed grades in this section. For example, $\la U\widetilde{U} \ra_0$ is denoted by $H_0$.

In the following theorem, we use the operation $\va$ (see the details in Section \ref{sec:2})
$$U^{\va}:=U^{\va_3}=U_0+U_1+U_2+U_3-U_4-U_5-U_6-U_7+U_8+\cdots$$

\begin{theorem}\label{thNF} For the cases $n\leq 5$, there exist the following functionals $N:\cl_{p,q}\to\cl^0_{p,q}$:
\begin{eqnarray*}
N(U)&:=&U\widehat{U}=\widehat{U}U, \qquad n=1;\\
N(U)&:=&U \widehat{\widetilde{U}}=\widehat{\widetilde{U}}U=\widehat{U}\widetilde{U}=\widetilde{U}\widehat{U}, \qquad n=2;\\
N(U)&:=& U \widetilde{U} \widehat{U} \widehat{\widetilde{U}}=U \widehat{\widetilde{U}}\widehat{U} \widetilde{U}=U \widehat{U} \widetilde{U} \widehat{\widetilde{U}}=U \widehat{\widetilde{U}} \widetilde{U} \widehat{U}=\widehat{U}\widetilde{U}U\widehat{\widetilde{U}}= \widehat{U}\widetilde{U}\widehat{\widetilde{U}}U=\widehat{U}U \widehat{\widetilde{U}}\widetilde{U}=\widehat{U}\widehat{\widetilde{U}}U\widetilde{U}\\
&=&\widetilde{U}\widehat{U}U\widehat{\widetilde{U}}=\widetilde{U}\widehat{U}\widehat{\widetilde{U}}U=\widetilde{U}U\widehat{\widetilde{U}}\widehat{U}= \widetilde{U}\widehat{\widetilde{U}}U\widehat{U}=\widehat{\widetilde{U}}U\widehat{U}\widetilde{U}=\widehat{\widetilde{U}}U\widetilde{U}\widehat{U}= \widehat{\widetilde{U}} \widehat{U}\widetilde{U}U=\widehat{\widetilde{U}}\widetilde{U}\widehat{U}U,\qquad n=3;\\
N(U)&:=&U \widetilde{U} (\widehat{U} \widehat{\widetilde{U}})^{\va}=U\widehat{\widetilde{U}}(\widehat{U} \widetilde{U})^{\va}=
\widetilde{U} (\widehat{U} \widehat{\widetilde{U}})^{\va}U=\widehat{\widetilde{U}}(\widehat{U} \widetilde{U})^{\va} U=(\widehat{U}\widetilde{\widehat{U}})^\va U\widetilde{U}=(\widehat{U}\widetilde{U})^\va U \widehat{\widetilde{U}}\\
&=&\widehat{U}\widehat{\widetilde{U}}(U\widetilde{U})^\va=\widehat{\widetilde{U}}(U\widetilde{U})^\va\widehat{U}= (U\widetilde{U})^\va \widehat{U} \widehat{\widetilde{U}}=\widehat{U}\widetilde{U} (U \widehat{\widetilde{U}})^\va=(U \widehat{\widetilde{U}})^\va \widehat{U} \widetilde{U}= \widetilde{U} (U \widehat{\widetilde{U}})^\va \widehat{U}\\
&=&U(\widehat{\widetilde{U}} \widehat{U})^\va \widetilde{U}=U(\widetilde{U}\widehat{U})^\va \widehat{\widetilde{U}}= (\widehat{\widetilde{U}}\widehat{U})^\va \widetilde{U}U=(\widetilde{U}\widehat{U})^\va\widehat{\widetilde{U}}U= \widetilde{U} U(\widehat{\widetilde{U}}\widehat{U})^\va=\widehat{\widetilde{U}}U(\widetilde{U}\widehat{U})^\va\\
&=&\widehat{U} (\widetilde{U} U)^\va \widehat{\widetilde{U}}=(\widetilde{U} U)^\va \widehat{\widetilde{U}}\widehat{U}= \widehat{\widetilde{U}}\widehat{U}(\widetilde{U} U)^\va=\widehat{U}(\widehat{\widetilde{U}} U)^\va \widetilde{U} =(\widehat{\widetilde{U}}U)^\va\widetilde{U}\widehat{U}=\widetilde{U}\widehat{U}(\widehat{\widetilde{U}} U)^\va,\qquad n=4;\\
N(U)&:=&U\widetilde{U} (\widehat{U} \widehat{\widetilde{U}})^{\va}(U\widetilde{U} (\widehat{U} \widehat{\widetilde{U}})^{\va})^{\va}=
U\widehat{\widetilde{U}}\widehat{U}\widetilde{U}(\widehat{U} \widetilde{U}U\widehat{\widetilde{U}})^{\va},\footnotemark\qquad  n=5.
\end{eqnarray*}
\footnotetext{And more than 400 other formulas obtained from the presented here two formulas: we can take the reversion, the grade involution, or the Clifford conjugation of the scalar $N(U)$; we can do cyclic permutations of multipliers in the obtained products because the left inverse equals to the right inverse; we can use the properties of the operation $\va$ (\ref{T5}) and (\ref{T6}); also  we can use $N(U)=N(\widehat{U})=N(\widetilde{U})$ because of the results of Section \ref{sec:4}. We do not present all these formulas here because of their large number.}
As a consequence, if $N(U)\neq 0$, then there exists $U^{-1}$ with the following explicit form
\begin{equation}
U^{-1}=\frac{F(U)}{N(U)},\qquad F(U):=\left\lbrace
\begin{array}{ll}
\widehat{U}, & \mbox{if $n=1$,}\\
\widehat{\widetilde{U}}, & \mbox{if $n=2$,}\\
\widetilde{U} \widehat{U} \widehat{\widetilde{U}}=\widehat{U}\widetilde{U}\widehat{\widetilde{U}}=\widehat{\widetilde{U}}\widehat{U}\widetilde{U}= \widehat{\widetilde{U}}\widetilde{U}\widehat{U}, & \mbox{if $n=3$,}\\
\widetilde{U} (\widehat{U} \widehat{\widetilde{U}})^{\va}=
\widehat{\widetilde{U}}(\widehat{U} \widetilde{U})^{\va}= (\widehat{\widetilde{U}}\widehat{U})^\va\widetilde{U}= (\widetilde{U}\widehat{U})^\va\widehat{\widetilde{U}}, & \mbox{if $n=4$,}\\
\widetilde{U} (\widehat{U} \widehat{\widetilde{U}})^{\va}(U\widetilde{U} (\widehat{U} \widehat{\widetilde{U}})^{\va})^{\va}= \widehat{\widetilde{U}}\widehat{U}\widetilde{U}(\widehat{U} \widetilde{U}U\widehat{\widetilde{U}})^{\va},\footnotemark & \mbox{if $n=5$.}
\end{array}
\right.
\end{equation}\footnotetext{And other formulas for $F(U)$ in the case $n=5$ because of the large number of different formulas for $N(U)$, see above.}
\end{theorem}
\begin{proof}
In the case $n=1$, we have $U\widehat{U}=U\widehat{\widetilde{U}}\in\cl^{\overline{0}}_{p,q}\oplus\cl^{\overline{3}}_{p,q}=\cl^0_{p,q}$ by Lemma \ref{HJ}. We have $U\widehat{U}=\widehat{U}U$ because the left inverse coincides with the right inverse.

In the case $n=2$, we have analogously $U\widehat{\widetilde{U}}\in\cl^{\overline{0}}_{p,q}\oplus\cl^{\overline{3}}_{p,q}=\cl^0_{p,q}$ by Lemma \ref{HJ}. Taking the grade involution, we get $(U\widehat{\widetilde{U}})\,\widehat\over=\widehat{U}\widetilde{U}$.

In the case $n=3$, the expression $\widehat{U}\widetilde{U}$ is invariant under the Clifford conjugation $(\widehat{U}\widetilde{U})\,\widehat{\widetilde\over}=\widehat{U}\widetilde{U}$ (we use the properties (\ref{grrev})), thus it lies in $\cl^{\overline{0}}_{p,q}\oplus\cl^{\overline{3}}_{p,q}=\cl^{0}_{p,q}\oplus\cl^{3}_{p,q}=\cen(\cl_{p,q})$, thus $U\widehat{U} \widetilde{U} \widetilde{\widehat{U}}=U \widetilde{\widehat{U}}\widehat{U}\widetilde{U}$. Using (\ref{grrev}), we conclude that this expression is invariant under the reversion and the Clifford conjugation
$$(U \widetilde{\widehat{U}}\widehat{U}\widetilde{U})\,\widetilde\over=U \widetilde{\widehat{U}}\widehat{U}\widetilde{U},\qquad (U\widehat{U} \widetilde{U} \widetilde{\widehat{U}})\,\widehat{\widetilde\over}=U\widehat{U} \widetilde{U} \widetilde{\widehat{U}},$$
thus it lies in $\cl^{\overline{0}}_{p,q}$ by (\ref{quat}), which coincides with $\cl^0_{p,q}$.\footnote{Let us give the alternative proof:
$U \widetilde{U} \widehat{U} \widehat{\widetilde{U}}=H\widehat{H}=(H_0+H_1)(H_0-H_1)= (H_0)^2-[ H_0, H_1]-(H_1)^2=(H_0)^2-(H_1)^2\in\cl^0_{p,q}$. One further alternative proof: $U \widehat{\widetilde{U}} \widehat{U} \widetilde{U}=J \widetilde{J}= (J_0+ J_3)( J_0-J_3)= (J_0)^2-[J_0, J_3]-(J_3)^2= (J_0)^2- (J_3)^2\in\cl^0_{p,q}.$} We obtain all other formulas for this case using $\widehat{U}\widetilde{U}=\widetilde{U}\widehat{U}\in\cl^{\overline{0}}_{p,q}\oplus\cl^{\overline{3}}_{p,q}=\cen(\cl_{p,q})$ and $U\widehat{\widetilde{U}}=\widetilde{\widehat{U}\widetilde{U}}=\widetilde{\widetilde{U}\widehat{U}}=\widehat{\widetilde{U}}U\in\cen(\cl_{p,q})$.

In the case $n=4$, using Lemma \ref{HJ} and (\ref{z2z2}) we have
\begin{eqnarray*}
U\widetilde{U} (\widehat{U} \widehat{\widetilde{U}})^{\va}&=&H \widehat{H}^\va=(H_0+H_1+H_4)(H_0-H_1-H_4)\\
&=&(H_0)^2-(H_1)^2-(H_4)^2-[H_0, H_1+H_4]-\{H_1, H_4\}=(H_0)^2-(H_1)^2-(H_4)^2\in\cl^0_{p,q},
\end{eqnarray*}
where $\{H_1, H_4\}=0$ because $e_{123 \ldots n}$ anticommutes with odd elements in the case of even $n$. We obtain the second formula using
\begin{eqnarray*}
U\widehat{\widetilde{U}}(\widehat{U}\widetilde{U})^{\va}&=&J \widehat{J}^\va=(J_0+J_3+J_4)(J_0-J_3-J_4)\\
&=&(J_0)^2-(J_3)^2-(J_4)^2-[J_0, J_3+ J_4]-\{J_3, J_4\}=(J_0)^2-(J_3)^2- (J_4)^2\in\cl^0_{p,q},
\end{eqnarray*}
where $(J_3)^2\in\cl^0_{p,q}$ because $J_3=e_{1234} W_1$ for some element $W_1\in\cl^1_{p,q}$. We get the other formulas by taking the reversion, the grade involution, or the Clifford conjugation of the scalar $N(U)$, doing cyclic permutations of multipliers in the obtained products (we can do this because the left inverse equals to the right inverse), and using the properties (\ref{T5}) and (\ref{T6}).

In the case $n=5$, we have
\begin{eqnarray*}
Y&:=&U\widetilde{U} (\widehat{U} \widehat{\widetilde{U}})^{\va}=H \widehat{H}^\va=(H_0+H_1+H_4+H_5)(H_0-H_1-H_4+H_5)\\
&=&(H_0)^2- (H_1)^2- (H_4)^2+(H_5)^2-[H_0, H_1+H_4]+\{H_0, H_5\}-\{H_1, H_4\}+[H_1+H_4, H_5]\\
&=&(H_0)^2-(H_1)^2-(H_4)^2+(H_5)^2 +2 H_0 H_5-\{H_1, H_4\}\in\cl^0_{p,q}\oplus\cl^5_{p,q},
\end{eqnarray*}
where $\{H_1, H_4\}\in\cl^5_{p,q}$, because it lies in $\cl^{\overline{1}}_{p,q}$ by (\ref{z2z2}) and the grade can be only 3 and 5. Finally,
\begin{eqnarray*}
Y Y^{\va}&=&(Y_0+Y_5)(Y_0-Y_5)=(Y_0)^2-(Y_5)^2-[Y_0, Y_5]=(Y_0)^2-(Y_5)^2\in\cl^0_{p,q}.
\end{eqnarray*}
We obtain the second formula using
\begin{eqnarray*}
Z\!&:=&\!U\widehat{\widetilde{U}}\widehat{U}\widetilde{U}=J \widehat{J}=(J_0+J_3+J_4)(J_0-J_3+J_4)=(J_0)^2-(J_3)^2+(J_4)^2-[J_0, J_3]+\{J_0, J_4\}+[ J_3, J_4]\\
\!&=&\!(J_0)^2-(J_3)^2+(J_4)^2+2 J_0 J_4+[J_3, J_4]\in\cl^0_{p,q}\oplus\cl^1_{p,q}\oplus\cl^4_{p,q},
\end{eqnarray*}
because $(J_3)^2=\frac{1}{2}\{J_3, J_3\}\in\cl^{\overline{0}}_{p,q}=\cl^0_{p,q}\oplus\cl^4_{p,q}$ and $[J_3, J_4]\in\cl^1_{p,q}$ by (\ref{z2z2}).
Finally,
\begin{eqnarray*}
Z \widehat{Z}^{\va}&=&(Z_0+Z_1+Z_4)(Z_0-Z_1-Z_4)=(Z_0)^2-(Z_1)^2-(Z_4)^2-[Z_0, Z_1+Z_4]-\{Z_1, Z_4\}\\
&=&(Z_0)^2-( Z_1)^2-(Z_4)^2\in\cl^0_{p,q},\quad \mbox{where}\\
\{Z_1, Z_4\}&=&\{[J_3, J_4], 2J_0 J_4-\la (J_3)^2\ra_4 \}=2 J_0 \{[J_3, J_4], J_4\}-\{[J_3, J_4], \la (J_3)^2\ra_4 \}=0,
\end{eqnarray*}
because $\{[J_3, J_4], J_4\}=[J_3, J_4^2]=0$, $J_4^2\in\cl^0_{p,q}$, and
$$\{[J_3, J_4], \la (J_3)^2 \ra_4 \}=\la \{[J_3, J_4], \la (J_3)^2\ra_4 \} \ra_5=-2 \la J_4 [ J_3, \la (J_3)^2 \ra_4 ]\ra_5=-2\la J_4 [J_3, (J_3)^2-\la (J_3)^2\ra_0 ]\ra_5=0,$$
where we used $\la UV \ra_n=\la VU \ra_n$ for odd $n$ (see Lemma \ref{lem1}). We get the other formulas by taking the reversion, the grade involution, or the Clifford conjugation of the scalar $N(U)$, doing cyclic permutations of multipliers in the obtained products, and using the properties (\ref{T5}) and (\ref{T6}).
\end{proof}

In \cite{acus}, there are 92 formulas (20 formulas in the form of doublets and 72 formulas in the form of triplets, see Tables 4 and 5 in \cite{acus}) for the determinant in the case $n=6$. They were obtained by computer calculations. Let us present an analytical proof that all these formulas are equal to (\ref{nu6}), where we use only three operations of conjugation $\,\widehat\over\,$, $\,\widetilde\over\,$, $\va$, or (\ref{nu62}), where we use two operations $\,\widetilde\over\,$ and $\stackrel{\over\quad}{\quad}$.

\begin{lemma}\label{lemn6} For $n=6$, there exists the following functional $N:\cl_{p,q}\to\cl^0_{p,q}$:
$$N(U)=\frac{1}{3}(A+2B),$$
where
$$A=H\widehat{H} (\widehat{H}H)^\va,\qquad B=H(\widehat{H}^\va(\widehat{H}^\va H^\va)^\va)^\va=H(({H}^\va \widehat{H}^\va)^\va\widehat{H}^\va)^\va,\qquad H=U\widetilde{U}.$$
Substituting $A$, $B$, and $H$, we get
\begin{eqnarray}
N(U)&=&\frac{1}{3}U \widetilde{U} \widehat{U} \widehat{\widetilde{U}}(\widehat{U}\widehat{\widetilde{U}}U\widetilde{U})^\va+ \frac{2}{3}U\widetilde{U}((\widehat{U}\widehat{\widetilde{U}})^\va((\widehat{U}\widehat{\widetilde{U}})^\va(U\widetilde{U})^\va)^\va)^\va.\label{nu6}
\end{eqnarray}
If $N(U)\neq 0$, then there exists
\begin{eqnarray}
U^{-1}&=&\frac{1}{N(U)}(\frac{1}{3}\widetilde{U} \widehat{U} \widehat{\widetilde{U}}(\widehat{U}\widehat{\widetilde{U}}U\widetilde{U})^\va+ \frac{2}{3}\widetilde{U}((\widehat{U}\widehat{\widetilde{U}})^\va((\widehat{U}\widehat{\widetilde{U}})^\va(U\widetilde{U})^\va)^\va)^\va).
\end{eqnarray}
\end{lemma}
\begin{proof} We use the formula from \cite{acus}, which is obtained by computer calculations:
$$\frac{1}{3}H H_{\underline{1}, \underline{5}} (H_{\underline{1}, \underline{5}} H)_{\underline{4}}+ \frac{2}{3} H (H_{\underline{4}, \underline{5}}(H_{\underline{4}, \underline{5}} H_{\underline{1}, \underline{4}})_{\underline{4}})_{\underline{1}, \underline{4}}\in\cl^0_{p,q},\qquad H=U\widetilde{U},$$
where we denote grade-negation operations by\footnote{In this paper, we denote the grade-negation operations by $\underline{k}$ (not by $\overline{k}$ as in \cite{acus}) to avoid confusion with the notation of quaternion types.} $U_{\underline{k}}:=U-2\la U \ra_k$ and $U_{\underline{k}, \underline{l}}:=U-2\la U \ra_k-2 \la U\ra_l$. By Lemma~\ref{HJ}, we get $H=\widetilde{H}\in\cl^{\overline{0}}_{p,q}\oplus\cl^{\overline{1}}_{p,q}= \cl^0_{p,q}\oplus\cl^1_{p,q}\oplus\cl^4_{p,q}\oplus\cl^5_{p,q}$. Thus
\begin{eqnarray*}
H_{\underline{4}, \underline{5}}=H^\va,\quad H_{\underline{1}, \underline{4}}=\widehat{H}^\va,\quad H_{\underline{1}, \underline{5}}=\widehat{H}.
\end{eqnarray*}
Using (\ref{grrev}) and (\ref{quat}), we get
\begin{eqnarray*}
&&(H \widehat{H})\,\widehat{\widetilde\over}=\widetilde{H} \widehat{\widetilde{H}}=H\widehat{H},\qquad (H^\va \widehat{H}^\va)\,\widehat{\widetilde\over}=\widetilde{H}^\va \widehat{\widetilde{H}}^\va=H^\va \widehat{H}^\va,\\
&&H\widehat{H},\quad  H^\va \widehat{H}^\va\in\cl^{\overline{0}}_{p,q}\oplus\cl^{\overline{3}}_{p,q}= \cl^0_{p,q}\oplus\cl^3_{p,q}\oplus\cl^4_{p,q},\quad (H\widehat{H})_{\underline{4}}=(H\widehat{H})^\va,\quad (H^\va \widehat{H}^\va)_{\underline{4}}=(H^\va\widehat{H}^\va)^\va.
\end{eqnarray*}
Using Lemma \ref{HJ}, the properties (\ref{grrev}), (\ref{T5}), and (\ref{quat}), we get
\begin{eqnarray*}
&&(H (H \widehat{H})^\va)\,\widetilde\over=(\widehat{H} H)^\va H=H (H \widehat{H})^\va,\qquad (H^\va (H^\va \widehat{H}^\va)^\va)\,\widetilde\over=(\widehat{H}^\va H^\va)^\va H^\va=H^\va (H^\va \widehat{H}^\va)^\va,\\
&&H (H \widehat{H})^\va,\qquad H^\va (H^\va \widehat{H}^\va)^\va\in \cl^{\overline{0}}_{p,q}\oplus\cl^{\overline{1}}_{p,q}= \cl^0_{p,q}\oplus\cl^1_{p,q}\oplus\cl^4_{p,q}\oplus\cl^5_{p,q},\\
&&(H (H \widehat{H})^\va)_{\underline{1}, \underline{5}}= (H (H \widehat{H})^\va)\,\widehat\over,\qquad  (H^\va (H^\va \widehat{H}^\va)^\va)_{\underline{1}, \underline{5}}= (H^\va (H^\va \widehat{H}^\va)^\va)\,\widehat\over.
\end{eqnarray*}
Finally, we obtain (\ref{nu6}).
\end{proof}

\begin{lemma}\label{lemn62} For $n=6$, there exists the following functional $N:\cl_{p,q}\to\cl^0_{p,q}$:
$$N(U)=\frac{1}{3}(C+2D),$$
where
$$C=H H \overline{(H H)},\qquad D=H \overline{(\overline{H} \overline{(\overline{H}\, \overline{H}) })}=H \overline{(\overline{(\overline{H}\, \overline{H})}\overline{H})},\qquad H=U \widetilde{U}.$$
Substituting $C$, $D$, and $H$, we get
\begin{eqnarray}
N(U)=\frac{1}{3}U \widetilde{U} U \widetilde{U} \overline{U\widetilde{U} U\widetilde{U}}+\frac{2}{3}U\widetilde{U} \overline{(\overline{(U\widetilde{U})}\overline{(\overline{U \widetilde{U}}\, \overline{U \widetilde{U}})})}.\label{nu62}
\end{eqnarray}
If $N(U)\neq 0$, then there exists
\begin{eqnarray}
U^{-1}&=&\frac{1}{N(U)}(\frac{1}{3}\widetilde{U} U \widetilde{U} \overline{U\widetilde{U} U\widetilde{U}}+\frac{2}{3}\widetilde{U} \overline{(\overline{(U\widetilde{U})}\overline{(\overline{U \widetilde{U}}\, \overline{U \widetilde{U}})})}).
\end{eqnarray}
\end{lemma}
\begin{proof} We use the other formula from \cite{acus}, which is obtained by computer calculations:
$$\frac{1}{3}H H (H H)_{\underline{1}, \underline{4}, \underline{5}}+ \frac{2}{3} H (H_{\underline{1}, \underline{4}, \underline{5}} (H_{\underline{1}, \underline{4}, \underline{5}} H_{\underline{1}, \underline{4}, \underline{5}})_{\underline{1}, \underline{4}, \underline{5}})_{\underline{1}, \underline{4}, \underline{5}}\in\cl^0_{p,q},\qquad H:=U\widetilde{U}.$$
Since $H\in\cl^{\overline{0}}_{p,q}\oplus\cl^{\overline{1}}_{p,q}= \cl^0_{p,q}\oplus\cl^1_{p,q}\oplus\cl^4_{p,q}\oplus\cl^5_{p,q}$, we have $H_{\underline{1}, \underline{4}, \underline{5}}=\overline{H}$.
Since
\begin{eqnarray*}
\widetilde{H}=H,\qquad \widetilde{HH}=HH,\qquad \widetilde{\overline{H}\, \overline{H}}=\overline{H}\, \overline{H},
\end{eqnarray*}
we have
\begin{eqnarray*}
&&HH,\quad \overline{H}\, \overline{H}\in \cl^{\overline{0}}_{p,q}\oplus\cl^{\overline{1}}_{p,q}= \cl^0_{p,q}\oplus\cl^1_{p,q}\oplus\cl^4_{p,q}\oplus\cl^5_{p,q},\qquad (H^2)_{\underline{1}, \underline{4}, \underline{5}}= \overline{H^2},\qquad (\overline{H}\, \overline{H})_{\underline{1}, \underline{4}, \underline{5}}=\overline{\overline{H} \, \overline{H}}.
\end{eqnarray*}
Using Lemma \ref{lemmaover}, we get
\begin{eqnarray*}
&&(\overline{H}\, \overline{(\overline{H}\,\overline{H})})\,\widetilde\over=\overline{(\overline{H}
, \overline{H})}\,\overline{H}=\overline{H}\, \overline{(\overline{H}\, \overline{H})},\\
&&\overline{H}\, \overline{(\overline{H}\, \overline{H})}\in\cl^{\overline{0}}_{p,q}\oplus\cl^{\overline{1}}_{p,q}= \cl^0_{p,q}\oplus\cl^1_{p,q}\oplus\cl^4_{p,q}\oplus\cl^5_{p,q},\qquad (\overline{H}\, \overline{(\overline{H}\, \overline{H})})_{\underline{1}, \underline{4}, \underline{5}}=\overline{(\overline{H}\, \overline{(\overline{H}\, \overline{H})})}.
\end{eqnarray*}
\end{proof}

The other formulas from \cite{acus} (with doublets and triplets) coincide with (\ref{nu6}) or (\ref{nu62}) because of the properties of the grade involution (\ref{grrev}), the properties of the operation $\va$ (\ref{T5}), and the properties of the operation $\stackrel{\over\quad}{\quad}$ (see Lemma \ref{lemmaover}). For example,
\begin{eqnarray*}
&&\frac{1}{3}H(H_{\underline{4}}(H_{\underline{4}}H_{\underline{4}})_{\underline{1}, \underline{4}, \underline{5}})_{\underline{4}}+ \frac{1}{3}H((H_{\underline{4}}H_{\underline{4}})_{\underline{4}}H_{\underline{1}, \underline{4}, \underline{5}})_{\underline{1}, \underline{4}, \underline{5}} +\frac{1}{3}HH(HH)_{\underline{1}, \underline{4}, \underline{5}}\\
&&=\frac{1}{3}H\overline{(\overline{\widehat{H}}\overline{(\overline{\widehat{H}}\, \overline{\widehat{H}})})\,\widehat\over}+ \frac{1}{3}H\overline{(\overline{(\overline{\widehat{H}}\, \overline{\widehat{H}})}\,\widehat\over\,\,\overline{H})} +\frac{1}{3}HH\overline{(HH)}=\frac{1}{3}HH\overline{(HH)}+\frac{2}{3}H \overline{(\overline{H}\, \overline{(\overline{H}\, \overline{H}) })}.
\end{eqnarray*}

From the computer calculations \cite{acus}, it follows that the expressions (\ref{nu6}) and (\ref{nu62}) coincide too. In our terms, this means that if we represent the operation $\stackrel{\over\quad}{\quad}$ in (\ref{nu62}) as a linear combination of the other operations of conjugation using (\ref{over6}), then (\ref{nu62}) should coincide with (\ref{nu6}). However it is difficult to give an analytical proof of this fact because of cumbersomeness of the calculations and nontrivial properties of the operations $\va$ and~$\stackrel{\over\quad}{\quad}$.

Note that in the cases of $n\leq 5$, the formulas from Theorem \ref{thNF} can be rewritten in the following form using the operation $\stackrel{\over\quad}{\quad}$ instead of the operation $\va$:
\begin{eqnarray}
N(U)&=&J,\qquad n=1, 2;\nonumber\\
N(U)&=&J \overline{J}=H\overline{H},\qquad n=3, 4;\label{htb}\\
N(U)&=&J \widehat{J} \overline{J \widehat{J}},\qquad n=5.\nonumber
\end{eqnarray}

In Theorem \ref{thNF}, we have 16 different expressions for $N(U)$ in the case $n=3$. All of them are the products of the four elements $U$, $\widetilde{U}$, $\widehat{U}$, $\widehat{\widetilde{U}}$ in a different order. We have $4!=24$ different permutations of 4 elements. It can be proved that the remaining 8 expressions
$$U\widetilde{U} \widehat{\widetilde{U}} \widehat{U},\quad \widetilde{U} U \widehat{U} \widehat{\widetilde{U}},\quad U\widehat{U}\widehat{\widetilde{U}}\widetilde{U},\quad \widehat{U}U\widetilde{U}\widehat{\widetilde{U}},\quad  \widehat{\widetilde{U}}\widetilde{U} U \widehat{U},\quad \widetilde{U}\widehat{\widetilde{U}}\widehat{U}U,\quad  \widehat{U}\widehat{\widetilde{U}}\widetilde{U}U,\quad \widehat{\widetilde{U}}\widehat{U} U\widetilde{U}$$
are not elements of grade $0$. However, their linear combinations are elements of grade $0$ (see the next lemma).

\begin{lemma}\label{lem24} In the case $n=3$, we have
\begin{eqnarray}
U\widetilde{U} \widehat{\widetilde{U}} \widehat{U}+\widetilde{U} U \widehat{U} \widehat{\widetilde{U}}= U\widehat{U}\widehat{\widetilde{U}}\widetilde{U}+\widehat{U}U\widetilde{U}\widehat{\widetilde{U}}= \widehat{\widetilde{U}}\widetilde{U} U \widehat{U}+\widetilde{U}\widehat{\widetilde{U}}\widehat{U}U= \widehat{U}\widehat{\widetilde{U}}\widetilde{U}U+ \widehat{\widetilde{U}}\widehat{U} U\widetilde{U}\in\cl^0_{p,q}.\label{yy1}
\end{eqnarray}
\end{lemma}
\begin{proof} Using (\ref{grrev}), we can verify that the grade involution and the reversion do not change the following two expressions
\begin{eqnarray*}
&&(U\widehat{U}\widehat{\widetilde{U}}\widetilde{U}+\widehat{U}U\widetilde{U}\widehat{\widetilde{U}})\,\widehat\over= U\widehat{U}\widehat{\widetilde{U}}\widetilde{U}+\widehat{U}U\widetilde{U}\widehat{\widetilde{U}},\qquad (U\widehat{U}\widehat{\widetilde{U}}\widetilde{U}+\widehat{U}U\widetilde{U}\widehat{\widetilde{U}})\,\widetilde\over= U\widehat{U}\widehat{\widetilde{U}}\widetilde{U}+\widehat{U}U\widetilde{U}\widehat{\widetilde{U}},\\
&&(\widehat{\widetilde{U}}\widetilde{U} U \widehat{U}+\widetilde{U}\widehat{\widetilde{U}}\widehat{U}U)\,\widehat\over= \widehat{\widetilde{U}}\widetilde{U} U \widehat{U}+\widetilde{U}\widehat{\widetilde{U}}\widehat{U}U,\qquad (\widehat{\widetilde{U}}\widetilde{U} U \widehat{U}+\widetilde{U}\widehat{\widetilde{U}}\widehat{U}U)\,\widetilde\over= \widehat{\widetilde{U}}\widetilde{U} U \widehat{U}+\widetilde{U}\widehat{\widetilde{U}}\widehat{U}U,
\end{eqnarray*}
Using (\ref{quat}), we conclude that $U\widehat{U}\widehat{\widetilde{U}}\widetilde{U}+\widehat{U}U\widetilde{U}\widehat{\widetilde{U}}$ and $\widehat{\widetilde{U}}\widetilde{U} U \widehat{U}+\widetilde{U}\widehat{\widetilde{U}}\widehat{U}U$ belong to $\cl^{\overline{0}}_{p,q}=\cl^0_{p,q}$. Using the properties $\la UV \ra_0=\la VU \ra_0$, $\la U+ V \ra_0= \la U \ra_0 +\la V \ra_0$, we get
\begin{eqnarray*}
U\widehat{U}\widehat{\widetilde{U}}\widetilde{U}+\widehat{U}U\widetilde{U}\widehat{\widetilde{U}}= \la U\widehat{U}\widehat{\widetilde{U}}\widetilde{U}+\widehat{U}U\widetilde{U}\widehat{\widetilde{U}} \ra_0= \la \widehat{\widetilde{U}}\widetilde{U} U \widehat{U}+\widetilde{U}\widehat{\widetilde{U}}\widehat{U}U \ra_0 = \widehat{\widetilde{U}}\widetilde{U} U \widehat{U}+\widetilde{U}\widehat{\widetilde{U}}\widehat{U}U.
\end{eqnarray*}
Also we can verify that the Clifford conjugation (superposition of the grade involution and the reversion) does not change the following two expressions
\begin{eqnarray*}
&&(U\widetilde{U} \widehat{\widetilde{U}} \widehat{U}+\widetilde{U} U \widehat{U} \widehat{\widetilde{U}})\,\widehat{\widetilde\over}= U\widetilde{U} \widehat{\widetilde{U}} \widehat{U}+\widetilde{U} U \widehat{U} \widehat{\widetilde{U}},\qquad (\widehat{U}\widehat{\widetilde{U}}\widetilde{U}U+ \widehat{\widetilde{U}}\widehat{U} U\widetilde{U})\,\widehat{\widetilde\over}= \widehat{U}\widehat{\widetilde{U}}\widetilde{U}U+ \widehat{\widetilde{U}}\widehat{U} U\widetilde{U}.
\end{eqnarray*}
Using (\ref{quat}), we conclude that $U\widetilde{U} \widehat{\widetilde{U}} \widehat{U}+\widetilde{U} U \widehat{U} \widehat{\widetilde{U}}$ and $\widehat{U}\widehat{\widetilde{U}}\widetilde{U}U+ \widehat{\widetilde{U}}\widehat{U} U\widetilde{U}$
belong to the center $\cl^{\overline{0}}_{p,q}\oplus\cl^{\overline{3}}_{p,q}=\cl^{0}_{p,q}\oplus\cl^3_{p,q}= \cen(\cl_{p,q})$.
Using the properties $\la UV \ra_{\cen}=\la VU \ra_{\cen}$, $\la U+V \ra_{\cen}=\la U \ra_{\cen}+ \la V \ra_{\cen}$, we get
\begin{eqnarray*}
&&U\widetilde{U} \widehat{\widetilde{U}} \widehat{U}+\widetilde{U} U \widehat{U} \widehat{\widetilde{U}}=\la U\widetilde{U} \widehat{\widetilde{U}} \widehat{U}+\widetilde{U} U \widehat{U} \widehat{\widetilde{U}} \ra_{\cen}=
\la \widehat{\widetilde{U}}\widetilde{U} U \widehat{U}+\widetilde{U}\widehat{\widetilde{U}}\widehat{U}U \ra_{\cen}= \widehat{\widetilde{U}}\widetilde{U} U \widehat{U}+\widetilde{U}\widehat{\widetilde{U}}\widehat{U}U,\\
&&\widehat{U}\widehat{\widetilde{U}}\widetilde{U}U+ \widehat{\widetilde{U}}\widehat{U} U\widetilde{U}= \la \widehat{U}\widehat{\widetilde{U}}\widetilde{U}U+ \widehat{\widetilde{U}}\widehat{U} U\widetilde{U} \ra_{\cen}= \la U\widehat{U}\widehat{\widetilde{U}}\widetilde{U}+\widehat{U}U\widetilde{U}\widehat{\widetilde{U}} \ra_{\cen}=U\widehat{U}\widehat{\widetilde{U}}\widetilde{U}+\widehat{U}U\widetilde{U}\widehat{\widetilde{U}}.
\end{eqnarray*}
Finally, all four expressions coincide and lie in $\cl^0_{p,q}$.
\end{proof}
The functionals (\ref{yy1}) are not of the (special) form $U F(U)$ (or $F(U) U$), so they can not be used to calculate the inverse of $U$, but they can be used for other purposes. There exist also other functionals in $\cl_{p,q}$ that are not of the special form. For example, the formulas for $\la U \ra_0$ from Theorem \ref{th10} give us such functionals. They are related to the trace of an element. In the next section, we consider other functionals, which are not always of the form $U F(U)$ (or $F(U) U$). They are other characteristic polynomial coefficients.

\section{Trace, determinant, and other characteristic polynomial coefficients in Clifford algebras}
\label{sec:4}

In this section, we introduce the concepts of characteristic polynomial coefficients (in particular, the trace and the determinant) in real Clifford algebras using matrix representations. Then we prove that these concepts do not depend on the choice of matrix representation and give alternative definitions of these concepts without using matrix representations and using only Clifford algebra operations. We present explicit formulas for the determinant, other characteristic polynomial coefficients, adjugate, and inverse in the case of arbitrary $n$ using only the operations of multiplication, summation, and operations of conjugation without explicit use of matrix representation.

We have the following isomorphisms between real Clifford algebras and matrix algebras
\begin{eqnarray}
\gamma:\cl_{p,q}&\to& L_{p,q}:=
\left\lbrace
\begin{array}{ll}
\Mat(2^{\frac{n}{2}},\R), & \mbox{if $p-q=0, 2 \mod 8$,}\\
\Mat(2^{\frac{n-1}{2}}, \R)\oplus\Mat(2^{\frac{n-1}{2}}, \R), & \mbox{if $p-q=1 \mod 8$,}\\
\Mat(2^{\frac{n-1}{2}}, \C), & \mbox{if $p-q=3, 7 \mod 8$,}\\
\Mat(2^{\frac{n-2}{2}}, \H), & \mbox{if $p-q=4, 6 \mod 8$,}\\
\Mat(2^{\frac{n-3}{2}}, \H)\oplus\Mat(2^{\frac{n-3}{2}}, \H), & \mbox{if $p-q=5 \mod 8$.}
\end{array}
\right.\label{gamma}
\end{eqnarray}
One can say that we have faithful representations $\gamma$ of the real Clifford algebras $\cl_{p,q}$ of the corresponding (minimal) dimensions over $\R$, $\R\oplus\R$, $\C$, $\H$, or $\H\oplus\H$ depending on $p-q\mod 8$.

Let us consider complexified Clifford algebras and the following isomorphisms to matrix algebras
\begin{eqnarray}
\beta:\C\otimes\cl_{p,q}&\to& M_{p,q}:=\left\lbrace
\begin{array}{ll}
\Mat(2^{\frac{n}{2}}, \C), & \mbox{if $n$ is even,}\\
\Mat(2^{\frac{n-1}{2}}, \C)\oplus\Mat(2^{\frac{n-1}{2}}, \C), & \mbox{if $n$ is odd.}
\end{array}
\right.\label{beta}
\end{eqnarray}
One can say that we have faithful representations $\beta$ of the complexified Clifford algebras $\C\otimes\cl_{p,q}$ of the corresponding (minimal) dimensions over $\C$ or $\C\oplus\C$ depending on $n\mod 2$.

We have $\cl_{p,q}\subset\C\otimes\cl_{p,q}$, and $\cl_{p,q}$ are isomorphic to some subalgebras of $M_{p,q}$. Thus we can consider the representation (of not minimal dimension)
\begin{eqnarray}
\beta: \cl_{p,q}\to \beta(\cl_{p,q})\subset M_{p,q}.\label{beta2}
\end{eqnarray}
This representation of $\cl_{p,q}$ is more useful for the problems of this paper than the representation $\gamma$ (\ref{gamma}), because it is more convenient for us to deal with complex matrices in the general case than with matrices over quaternions in some cases. Another reason for using the representation (\ref{beta2}) instead of the representation (\ref{gamma}) for the purposes of this paper is the structure of the formulas in Theorem \ref{thNF}. The formulas do not depend on $p$ and $q$, and the number of multipliers in the presented expressions equal $2^{[\frac{n+1}{2}]}$, i.e. depend on $n\mod 2$ and coincides with the dimension of the representation (\ref{beta2}), which we denote by
\begin{eqnarray}
N:=2^{[\frac{n+1}{2}]}.\label{N}
\end{eqnarray}
Let us present an explicit form of one of these representations ((\ref{beta}) for $\C\otimes\cl_{p,q}$ and (\ref{beta2}) for $\cl_{p,q}$). We denote this fixed representation by $\beta^\prime$. Let us consider the case $p=n$, $q=0$. To obtain the matrix representation for another signature with $q\neq 0$, we should multiply matrices $\beta^\prime(e_a)$, $a=p+1, \ldots, n$ by imaginary unit $i$. For the identity element, we always use the identity matrix $\beta^\prime(e)=I_N$ of the corresponding dimension $N$. We always take $\beta^\prime(e_{a_1 a_2\ldots  a_k})=\beta^\prime(e_{a_1})\beta^\prime(e_{a_2})\cdots \beta^\prime(e_{a_k})$. In the case $n=1$, we take $\beta^\prime(e_1)=\diag(1, -1)$. Suppose we know $\beta^\prime_a:=\beta^\prime(e_a)$, $a=1, \ldots, n$ for some fixed odd $n=2k+1$. Then for $n=2k+2$, we take the same $\beta^\prime(e_a)$, $a=1, \ldots, 2k+1$, and
$$
\beta^\prime(e_{2k+2})=\left(
                \begin{array}{cc}
                  0 & I_{\frac{N}{2}} \\
                  I_{\frac{N}{2}} & 0 \\
                \end{array}
              \right).
$$
For $n=2k+3$, we take
$$
\beta^\prime(e_a)=\left(
                \begin{array}{cc}
                  \beta^\prime_a & 0 \\
                  0 & -\beta^\prime_a \\
                \end{array}
              \right),\qquad a=1, \ldots, 2k+2,\qquad \beta^\prime(e_{2k+3})=\left(
                \begin{array}{cc}
                  i^{k+1} \beta^\prime_1 \cdots \beta^\prime_n & 0 \\
                  0 & -i^{k+1} \beta^\prime_1 \cdots \beta^\prime_n \\
                \end{array}
              \right).
$$
This recursive method gives us an explicit form of the matrix representation $\beta^\prime$ for all $n$.

By the following theorem, the projection onto the subspace of grade $0$ in $\cl_{p,q}$ coincides up to scalar with the trace of the corresponding matrix representation $\beta$ (\ref{beta2}).\footnote{Note that the same statement is valid for the matrix representation (\ref{beta}) $\beta:\C\otimes\cl_{p,q} \to M_{p,q}$ of the complexified Clifford algebra, see the details in \cite{rudn}.}
\begin{lemma}\label{rudn1} For the matrix representation $\beta$ (\ref{beta2}), we have
$$\frac{1}{N}\tr(\beta(U))=\la U \ra_0\in\cl^0_{p,q}.$$
\end{lemma}
\begin{proof} For the presented matrix representation $\beta^\prime$, we have $\tr(\beta^\prime(U))=\tr(\la U \ra_0 I_N)=N \la U \ra_0$ by construction. Let we have some another matrix representation $\beta$ of the same dimension. If $n$ is even, then using the Pauli theorem \cite{Pauli} (or using the representation theory) we conclude that there exists an element $T$ such that $\beta(e_a)=T^{-1} \beta^\prime(e_e) T$. We get $\beta(U)=T^{-1}\beta^\prime(U) T$ and $\tr(\beta(U))=\tr(\beta^\prime(U))$ using the property of trace. In the case odd $n$, by the Pauli theorem we can have also the relation $\beta(e_a)=-T^{-1} \beta^\prime(e_a) T$, which can be rewritten in the form $\beta(U)=T^{-1}\beta^\prime(\widehat{U}) T$ by linearity. We obtain $\tr(\beta(U))=\tr(\beta^\prime(\widehat{U}))=\tr(\beta^\prime(U))$, where we use $\la \widehat{U} \ra_0=\widehat{\la U \ra_0}=\la U\ra_0$.
\end{proof}

\begin{definition}
Let us introduce the concept of determinant $\Det(U)$ in the real Clifford algebra $\cl_{p,q}$ using the matrix representation $\beta$ (\ref{beta2}):\footnote{Note that if we will use the matrix representation $\gamma$ (\ref{gamma}) instead of the matrix representation $\beta$ (\ref{beta2}) in the definition of the determinant, then we obtain another concept of the determinant with values in $\R$, $\C$, or $\H$, which does not coincide with the first one in the general case. We need not this concept in this paper but use it, for example, in \cite{Lie3}.
In the cases $p-q=0, 1, 2 \mod 8$,  we can use the representation $\gamma$ (\ref{gamma}) and some fixed representation $\gamma^\prime$ (see the recursive method in \cite{Lie3}
or the method using idempotents and basis of the left ideal in~\cite{Abl}) instead of $\beta$ and $\beta^\prime$ and obtain the same concept of the determinant.}
\begin{eqnarray}
\Det(U):=\det(\beta(U))\in\cl^0_{p,q}\equiv \R,\qquad U\in\cl_{p,q}.\label{det}
\end{eqnarray}
\end{definition}
The determinant of the complex matrix is real in this case, because $\tr(\beta(U))$ is real for an arbitrary $U\in\cl_{p,q}$ (see Lemma \ref{rudn1}) and it is known from matrix theory that the determinant of a matrix is a real polynomial of traces of the matrix powers (see, for example, the Faddeev-LeVerrier algorithm for matrices).

Let us give one example. In the case $n=2$, $p=q=1$, for an arbitrary $U=ue+u_1 e_1+u_2 e_2+u_{12}e_{12}\in\cl_{p,q}$, $u, u_1, u_2, u_{12}\in\R$, we get the complex matrix
$$\beta^\prime(U)=\left(
                                     \begin{array}{cc}
                                       u+u_1 & iu_2+i u_{12} \\
                                       i u_2-iu_{12} & u-u_1 \\
                                     \end{array}
                                   \right)\in \Mat(2,\C)
$$
with the real trace $\tr(\beta^\prime(U))=2u\in\R$ and the real determinant $\det(\beta^\prime(U))=u^2-u_1^2+u_2^2-u_{12}^2\in\R$.

\begin{lemma}\label{rudn2} The determinant $\Det(U)$ (\ref{det}) is well-defined, i.e. it does not depend on the representation $\beta$~(\ref{beta2}).
\end{lemma}
\begin{proof}  Let us consider the representation $\beta^\prime$, which is discussed above. In the case of even $n$, for an arbitrary representation $\beta$ of the same dimension, by the Pauli theorem \cite{Pauli}, there exists $T$ such that $\beta(e_a)=T^{-1} \beta^\prime(e_a) T$. We get $\beta(U)=T^{-1}\beta^\prime(U) T$ and $\det(\beta(U))=\det(\beta^\prime(U))$. In the case of odd $n$, we can have also the relation $\beta(e_a)=-T^{-1} \beta^\prime(e_a) T$, which means $\beta(U)=T^{-1}\beta^\prime(\widehat{U}) T$ and $\det(\beta(U))=\det(\beta^\prime(\widehat{U})$.

Let us prove that $\det(\beta^\prime(\widehat{U}))=\det(\beta^\prime(U))$. For the representation $\beta^\prime$, we have $\beta^\prime(e_a)=\diag(\beta^\prime_a, -\beta^\prime_a)$, $a=1, \ldots, n$, where blocks $\beta^\prime_a$ and $\beta^\prime_a$ are identical up to sign. Thus the matrix  $\beta^\prime(e_{ab})=\beta^\prime(e_a) \beta^\prime(e_b)=\diag(\beta^\prime_a \beta^\prime_b, \beta^\prime_a \beta^\prime_b)$ has two identical blocks. We conclude that for the even part $\la U \ra_{\even}$ of the element $U$ we have $\beta^\prime(\la U \ra_{\even})=\diag(A, A)$ with two identical blocks $A$, and for the odd part $\la U \ra_{\odd}$ of the same element we have $\beta^\prime(\la U \ra_{\odd})=\diag(B, -B)$ with the two blocks $B$ and $-B$ differing in sign. Finally, we get $\beta^\prime(U)=\diag(A+B, A-B)$, $\beta^\prime(\widehat{U})=\diag(A-B, A+B)$, and $\det(\beta^\prime(U))=(A+B)(A-B)=\det(\beta^\prime(\widehat{U}))$.
\end{proof}

\begin{lemma}\label{rudn3} The operation $\Det: \cl_{p,q} \to \R$ has the following properties
\begin{eqnarray}
&&\Det(UV)=\Det(U)\Det(V),\qquad \Det(\lambda U)=\lambda^N \Det(U),\label{det1}\\
&&\Det(\widetilde{U})=\Det(\widehat{U})=\Det(U),\qquad \forall U, V\in\cl_{p,q},\qquad \forall \lambda\in\R;\label{det2}\\
&&\mbox{$U\in\cl_{p,q}\,$ is invertible if and only if $\,\Det(U)\neq 0$.}\label{det3}
\end{eqnarray}
As a consequence, we obtain
\begin{eqnarray}
\Det(T^{-1}UT)=\Det(U),\qquad \Det(T^{-1})=(\Det(T))^{-1},\qquad \forall U\in\cl_{p,q},\qquad \forall T\in\cl^\times_{p,q}.\label{det4}
\end{eqnarray}
\end{lemma}
\begin{proof} We get (\ref{det1}), (\ref{det3}), and (\ref{det4}) from the standard properties of the determinant of matrices.

Let us prove $\Det(\widehat{U})=\Det(U)$. If $n$ is even, then $\widehat{U}=(e_{1\ldots n})^{-1} U e_{1\ldots n}$, because $e_{1\ldots n}$ commutes with all even elements and anticommutes with all odd elements. We get $\Det(\widehat{U})=\Det((e_{1\ldots n})^{-1} U e_{1\ldots n})=\Det U$. In the case of odd $n$, we have already verified this for the representation $\beta^\prime$ in the proof of Lemma \ref{rudn2}. This is valid for an arbitrary representation $\beta$ (\ref{beta2}) because $\Det(U)$ does not depend on the choice of $\beta$.

Let us prove that $\Det(\widetilde{U})=\Det(U)$. We have the following relation between the transpose and the reversion or the Clifford conjugation (this depends on the matrix representation, see the details in \cite{Lie3}):
$$(\beta(U))^\T = \beta( e_{b_1 \ldots b_k} \widetilde{U} (e_{b_1 \ldots b_k})^{-1}),\qquad (\beta(U))^\T = \beta(e_{b_1 \ldots b_k} \widehat{\widetilde{U}} (e_{b_1 \ldots b_k})^{-1})$$
for some fixed basis element $e_{b_1 \ldots b_k}$. Finally, we get $\det(\beta(\widetilde{U}))=\det((\beta(U))^\T)=\det(\beta(U))$ and $\Det(\widetilde{U})=\Det(U)$.
\end{proof}

By Lemma \ref{rudn1}, we have a realization of the trace of an element $U\in\cl_{p,q}$ in terms of Clifford algebra operations without using matrix representations: $N \la U \ra_0$. Since $\Det(U)$ also does not depend on the representation $\beta$ by Lemma \ref{rudn2}, it would be an important task to find another definition (realization) of $\Det$ instead of the definition (\ref{det}) in terms of only Clifford algebra operations without using matrix representations. We do this using the relation between the determinant and the trace, which is known from matrix theory by the Cayley-Hamilton theory.

Let us give the example for $\cl_{p,q}$, $p+q=2$. Let we have an arbitrary element $U\in\cl_{p,q}$. For the complex matrix $A:=\beta(U)$ of dimension $2$, by the Cayley-Hamilton theorem, we have $A^2-\tr(A) A+ \det(A) I_2=0$. From this equation, we get $\det(A) I_2=A(\tr(A) I_2-A)$. Taking $\beta^{-1}$ and using $\tr(\beta(U))=2\la U\ra_0$, we obtain $\Det(U)= U(2\la U\ra_0-U)=U\widehat{\widetilde{U}}$, which coincides with $N(U)$ in Theorem \ref{thNF}. The expression $\Adj(U)=\widehat{\widetilde{U}}$ can be interpreted as the adjugate of the Clifford algebra element $U\in\cl_{p,q}$.

Now let us consider the general case.
\begin{definition}
Let we have $U\in\cl_{p,q}$. We call the characteristic polynomial of $U$
\begin{eqnarray}
\varphi_U(\lambda)&:=&\Det(\lambda e-U)=\det(\beta(\lambda e- U))=\det(\lambda I_N- \beta(U))\label{cpc}\\
&=&\lambda^N-C_{(1)} \lambda^{N-1}-\cdots-C_{{(N-1)}}\lambda -C_{(N)}\in\cl^0_{p,q},\nonumber
\end{eqnarray}
where\footnote{We use the notation with indices in round brackets ``$(k)$'' to avoid confusion with the notation  of the projection operations onto subspaces of fixed grades.} $C_{(j)}=C_{(j)}(U)\in\cl^0_{p,q}\equiv \R$, $j=1, \ldots, N$ can be interpreted as constants or as elements of grade $0$ and are called characteristic polynomial coefficients of $U$.
\end{definition}
We have $C_{(j)}(U)=c_{(j)} (\beta(U))$, where $c_{(j)} (\beta(U))$ are the ordinary characteristic polynomial coefficients of the matrix $\beta(U)$. By the Cayley-Hamilton theorem, we have
$$\varphi_U(U)=U^N-C_{(1)} U^{N-1}-\cdots-C_{{(N-1)}}U -C_{(N)}=0.$$
In particular, we have $C_{(N)}=(-1)^{N+1} \Det(U)=-\Det(U)$ (because $N=2^{[\frac{n+1}{2}]}$ is even) and $C_{(1)}=\tr(\beta(U))=N \la U\ra_0$.

\begin{lemma}\label{lemC}
We have $$C_{(k)}(\widehat{U})=C_{(k)}(\widetilde{U})=C_{(k)}(U),\qquad k=1, \ldots, N.$$
\end{lemma}
\begin{proof} We have this property for the $C_{(N)}(U)=-\Det(U)$ by  Lemma \ref{rudn3}. We get the same property for the other characteristic polynomial coefficients because of the definition of $C_{(k)}$ in (\ref{cpc}).
\end{proof}

We call the adjugate of an arbitrary Clifford algebra element $U\in\cl_{p,q}$ the element $\Adj(U)\in\cl_{p,q}$ such that
$$\Adj(U) U= U\Adj(U)=\Det(U).$$
There exists $$U^{-1}=\frac{\Adj(U)}{\Det(U)},$$
if and only if $\Det(U)\neq 0$. The expression $\Adj(U)$ is an analogue of the adjugate of matrix, namely $$\Adj(U)=\adj(\beta(U)).$$

\begin{theorem}\label{thLF} Let we have an arbitrary Clifford algebra element $U\in\cl_{p,q}$, $n=p+q$, and $N:=2^{[\frac{n+1}{2}]}$. Let us introduce the following set of Clifford algebra elements $U_{(k)}$, $k=1, \ldots, N$, and the set of scalars  $C_{(k)}\in\cl^0_{p,q}\equiv \R$, $k=1, \ldots, N$:
\begin{eqnarray}
U_{(1)}:=U,\qquad U_{(k+1)}:=U(U_{(k)}-C_{(k)}),\qquad C_{(k)}=\frac{N}{k} \la U_{(k)}\ra_0\in\cl^0_{p,q}\equiv \R.\label{Ck}
\end{eqnarray}
Then $C_{(k)}$ are the characteristic polynomial coefficients,
\begin{eqnarray}
\Det(U)=-U_{(N)}=-C_{(N)}=U(C_{(N-1)}-U_{(N-1)})\in\cl^0_{p,q}\equiv \R \label{det01}
\end{eqnarray}
is the determinant of $U$, and
\begin{eqnarray}
\Adj(U)=C_{(N-1)}-U_{{(N-1)}}\in\cl_{p,q}\label{adj01}
\end{eqnarray}
is the adjugate of $U$.

Alternatively, using the set of scalars
\begin{eqnarray}
S_{(k)}:= (-1)^{k-1}N (k-1)! \la U^k \ra_0\in\cl^0_{p,q}\equiv \R,\qquad k=1, \ldots, N,\label{Sk}
\end{eqnarray}
we have the following formulas
\begin{eqnarray}
&&C_{(k)}=\frac{(-1)^{k+1}}{k!}B_k(S_{(1)}, S_{(2)}, S_{(3)}, \ldots, S_{(k)}),\qquad k=1, \ldots, N,\label{Sk2}\\
&&\Det(U)=-C_{(N)}=\frac{1}{N!}B_N(S_{(1)}, S_{(2)}, S_{(3)}, \ldots, S_{(N)}),\label{det02}\\
&&\Adj(U)=\sum_{k=0}^{N-1} \frac{(-1)^{N+k-1}}{k!}U^{N-k-1}B_k(S_{(1)}, S_{(2)}, S_{(3)}, \ldots, S_{(k)}),\label{adj02}
\end{eqnarray}
where we use the complete Bell polynomials with the following two equivalent definitions
\begin{eqnarray*}
B_k(x_1, \ldots, x_k)&:=&\sum_{i=1}^k \sum \frac{k!}{j_1! j_2! \cdots j_{k-i+1}!} (\frac{x_1}{1!})^{j_1}(\frac{x_2}{2!})^{j_2}\cdots(\frac{x_{k-i+1}}{(k-i+1)!})^{j_{k-i+1}}\\
&=&\det \left(
             \begin{array}{ccccc}
               x_1 & C^1_{k-1} x_2  &  C^2_{k-1} x_3& \cdots & x_k \\
               -1 & x_1 &     C^1_{k-2} x_2 & \cdots & x_{k-1}\\
               0   & -1 & x_1 & \cdots & x_{k-2}\\
               \cdots & \cdots    &    \cdots        & \cdots & \cdots \\
               0 & 0    &  0  & \cdots & x_1 \\
             \end{array}
           \right),
\end{eqnarray*}
where the second sum is taken over all sequences $j_1, j_2, \ldots, j_{k-i+1}$ of nonnegative integers  satisfying the conditions $j_1+j_2+\cdots+j_{k-i+1}=i$ and $j_1+2j_2+3j_3+\cdots+(k-i+1)j_{k-i+1}=k$.
\end{theorem}

\begin{proof} The theorem follows from the Faddeev-LeVerrier algorithm (see, \cite{Gant, Hous}) for the matrix $\beta(U)\in M_{p,q}$, $U\in\cl_{p,q}$ and the techniques developed before the theorem (in particular, Lemmas \ref{rudn1}, \ref{rudn2}, and the generalizations of the concepts of the trace, determinant, and other characteristic polynomial coefficients to the case of Clifford algebras). The second part of the theorem follows from the method of calculating the characteristic polynomial coefficients using the complete Bell polynomials \cite{Bell}.
\end{proof}

Note that in Theorem \ref{thLF} all formulas use only operations in Clifford algebras and we need no matrix representations. We realize the trace, determinant, other characteristic polynomial coefficients, adjugate, and inverse using only the operations of summation, multiplication, and one operation of projection (the operation of projection $\la \quad \ra_0$ onto the subspace of grade $0$).

We can realize the operation $\la \quad \ra_0$ using the operation of conjugation $\overline{U}$ (see Section \ref{sec:2}):
\begin{eqnarray}
\la U \ra_0=\frac{U+\overline{U}}{2},\qquad \overline{U}=\la U \ra_0-\sum_{k=1}^n \la U \ra_k.\label{x2}
\end{eqnarray}
Substituting (\ref{x2}) into (\ref{Ck}) or (\ref{Sk}), we obtain the formulas for all characteristic polynomial coefficients, adjugate, and inverse using only the operation $\stackrel{\over\quad}{\quad}$. Let us write down explicit formulas in the case of small dimensions using (\ref{Ck}) and (\ref{x2}).

In the cases $n=1$ and $n=2$, we have:
\begin{eqnarray*}
&&C_{(1)}=2\la U \ra_0=U+\overline{U},\qquad \Det(U)=-C_{(2)}=-U_{(2)}=U(C_{(1)}-U)=U\overline{U},\\
&&\Adj(U)=\overline{U},\qquad U^{-1}=\frac{\Adj(U)}{\Det(U)}=\frac{\overline{U}}{U\overline{U}}.
\end{eqnarray*}
In the cases $n=3$ and $n=4$, we have:
\begin{eqnarray*}
&&C_{(1)}=4\la U \ra_0=2(U+\overline{U}),\qquad U_{(2)}=U(U-C_{(1)})=-(U^2+2U \overline{U}),\\
&&C_{(2)}=2\la U_{(2)} \ra_0=-(U^2+2 U \overline{U} + \overline{U^2}+ 2\overline{U \overline{U}}),\qquad
U_{(3)}=U(U_{(2)}-C_{(2)})=U(\overline{U^2}+2\overline{U \overline{U}}),\\
&&C_{(3)}=\frac{4}{3}\la U_{(3)}\ra_0=\frac{2}{3}(U\overline{U^2}+2U\overline{U \overline{U}}+\overline{U\overline{U^2}}+2\overline{U\overline{U \overline{U}}}),\\
&&\Det(U)=-C_{(4)}=-U_{(4)}=U(C_{(3)}-U_{(3)})=\frac{1}{3}U(-U\overline{U^2}-2U\overline{U \overline{U}}+2\overline{U\overline{U^2}}+4\overline{U\overline{U \overline{U}}}),\\
&&\Adj(U)=\frac{1}{3}(-U\overline{U^2}-2U\overline{U \overline{U}}+2\overline{U\overline{U^2}}+4\overline{U\overline{U \overline{U}}}),\\
&&U^{-1}=\frac{\Adj(U)}{\Det(U)}=\frac{(-U\overline{U^2}-2U\overline{U \overline{U}}+2\overline{U\overline{U^2}}+4\overline{U\overline{U \overline{U}}})}{U(-U\overline{U^2}-2U\overline{U \overline{U}}+2\overline{U\overline{U^2}}+4\overline{U\overline{U \overline{U}}})}.
\end{eqnarray*}
Alternatively, we can use the complete Bell polynomials. The complete Bell polynomials $B_k=B_k(x_1, \ldots, x_k)$ have the following explicit form for small $k\leq 8$:
\begin{eqnarray*}
&&B_1=x_1,\qquad B_2=x_1^2+x_2,\qquad B_3=x_1^3+3x_1x_2+x_3,\qquad B_4=x_1^4+6x_1^2x_2+4x_1x_3+3x_2^2+x_4,\\
&&B_5= x_1^5+10x_1^3x_2+15x_1x_2^2+10x_1^2x_3+10x_2x_3+5x_1x_4+x_5,\\
&&B_6=x_1^6+15x_1^4x_2+20x_1^3x_3+45x_1^2x_2^2+15x_2^3+60x_1x_2x_3+15x_1^2x_4+10x_3^2+15x_2x_4+6x_1x_5+x_6,\\
&&B_7=x_1^7+21x_1^5x_2+35x_1^4x_3+105x_1^3x_2^2+35x_1^3x_4+210x_1^2x_2x_3+105x_1x_2^3+21x_1^2x_5+105x_1x_2x_4\\
&&+ 70x_1x_3^2+105x_2^2x_3+7x_1x_6+21x_2x_5+35x_3x_4+x_7,\\
&&B_8=x_1^8+28x_1^6x_2+56x_1^5x_3+210x_1^4x_2^2+56x_1^5x_3+70x_1^4x_4+560x_1^3x_2x_3+420x_1^2x_2^3+56x_1^3x_5\\
&&+420x_1^2x_2x_4+280x_1^2x_3^2+840x_1x_2^2x_3+105x_2^4+28x_1^2x_6+168x_1x_2x_5+280x_1x_3x_4+210x_2^2x_4\\
&&+280x_2x_3^2+8x_1x_7+28x_2x_6+56x_3x_5+35x_4^2+x_8.
\end{eqnarray*}
Let us write down explicit formulas for the determinant in the cases of small dimensions using (\ref{det02}). For the cases $n=1$ and $n=2$, we get
$$\Det(U)=\frac{1}{2}B_2(2\la U\ra_0, -2\la U^2 \ra_0)=2\la U\ra_0^2-2\la U^2 \ra_0.$$
For the cases $n=3$ and $n=4$, we get
\begin{eqnarray*}
\Det(U)&=&\frac{1}{24}((4\la U \ra_0)^4+6(4\la U \ra_0)^2(-4\la U^2 \ra_0)+4(4\la U \ra_0)8\la U^3 \ra_0+3(-4\la U^2 \ra_0)^2+(-24\la U^4 \ra_0))\\
&=&\frac{1}{3}(32\la U \ra_0^4-48\la U \ra_0^2 \la U^2 \ra_0+16 \la U \ra_0\la U^3 \ra_0+6\la U^2 \ra_0^2-3\la U^4 \ra_0).
\end{eqnarray*}
For the cases $n=5$ and $n=6$, we get
\begin{eqnarray}
&&\Det(U)=\frac{1}{8!}((8\la U\ra_0)^8-28(8\la U\ra_0)^6 8\la U^2 \ra_0+56(8\la U\ra_0)^5 16\la U^3 \ra_0+210(8\la U\ra_0)^4(8\la U^2 \ra_0)^2\nonumber\\
&&+56(8\la U\ra_0)^5 16\la U^3 \ra_0-70(8 \la U \ra_0)^4 48 \la U^4 \ra_0-560(8\la U\ra_0)^3 8\la U^2 \ra_0 16\la U^3 \ra_0-420(8\la U\ra_0)^2(8\la U^2 \ra_0)^3\nonumber\\
&&+56(8\la U\ra_0)^3 192\la U^5\ra_0+420(8\la U\ra_0)^2 8\la U^2 \ra_0 48 \la U^4 \ra_0+280(8\la U\ra_0)^2(16\la U^3 \ra_0)^2\nonumber\\
&&+840(8\la U\ra_0)(8\la U^2 \ra_0)^2 16\la U^3 \ra_0+105(8\la U^2 \ra_0)^4-28(8\la U\ra_0)^2 960 \la U^6 \ra_0-168(8\la U\ra_0)8\la U^2 \ra_0 192\la U^5\ra_0\nonumber\\
&& -280(8\la U\ra_0)16\la U^3 \ra_0 48 \la U^4 \ra_0-210(8\la U^2 \ra_0)^2 48 \la U^4 \ra_0-280(8\la U^2 \ra_0)(16\la U^3 \ra_0)^2+8(8\la U\ra_0)5760 \la U^7 \ra_0\nonumber\\
&&+28(8\la U^2 \ra_0)960 \la U^6 \ra_0+56(16\la U^3 \ra_0)192\la U^5\ra_0+35(48 \la U^4 \ra_0)^2-40320 \la U^8 \ra_0).\nonumber
\end{eqnarray}
For the cases $n=7$ and $n=8$, the formula for the determinant of this type has 231 summands. Similarly we can write down explicit formulas for all characteristic polynomial coefficients in the case of arbitrary $n$. Also we can substitute (\ref{x2}) into these expressions and get the formulas using only the operation~$\stackrel{\over\quad}{\quad}$.

From the results of Section \ref{sec:2}, it follows that we can realize the operation $\la \quad \ra_0$ as a linear combination of the operations $\va_1$, $\va_2$, \ldots, $\va_m$, $m=[\log_2 n]+1$ and their superpositions (see Theorem \ref{th10}). For example, we can substitute the following expression (or other realizations of $\la \quad \ra_0$ from Theorem \ref{th10})
\begin{eqnarray}
\la U \ra_0=\frac{1}{2^m}(U+U^{\va_1}+U^{\va_2}+\cdots+U^{\va_1 \ldots \va_m}),\qquad m=[\log_2 n ]+1,\qquad U^{\va_1}=\widehat{U},\qquad U^{\va_2}=\widetilde{U}\label{x3}
\end{eqnarray}
into (\ref{Ck}) or (\ref{Sk}) and obtain explicit formulas for all characteristic polynomial coefficients (in particular the determinant), adjugate, and inverse using the operations $\va_1$, $\va_2$, \ldots, $\va_m$. We simplify the obtained formulas for the cases $n\leq 4$ (see the next theorem, the formulas for $C_{(2)}$ and $C_{(3)}$ in the cases $n=3, 4$ are new).

\begin{theorem}\label{LFA} In the case $n=1$, we have
\begin{eqnarray*}
C_{(1)}=U+\widehat{U}\in\cl^0_{p,q},\qquad \Det(U)=-C_{(2)}=U\widehat{U}\in\cl^0_{p,q},\qquad \Adj(U)=\widehat{U},\qquad U^{-1}=\frac{\widehat{U}}{\Det(U)}.
\end{eqnarray*}
In the case $n=2$, we have
\begin{eqnarray*}
C_{(1)}=U+\widehat{\widetilde{U}}\in\cl^0_{p,q},\qquad \Det(U)=-C_{(2)}=U\widehat{\widetilde{U}}\in\cl^0_{p,q},\qquad \Adj(U)=\widehat{\widetilde{U}},\qquad U^{-1}=\frac{\widehat{\widetilde{U}}}{\Det(U)}.
\end{eqnarray*}
In the case $n=3$, we have
\begin{eqnarray*}
&&C_{(1)}=U+\widehat{U}+\widetilde{U}+\widehat{\widetilde{U}}\in\cl^0_{p,q},\qquad C_{(2)}=-(U\widetilde{U}+U\widehat{U}+U\widehat{\widetilde{U}}+\widehat{U}\widehat{\widetilde{U}}+\widetilde{U}\widehat{\widetilde{U}} +\widehat{U}\widetilde{U})\in\cl^0_{p,q},\\
&&C_{(3)}=U\widehat{U}\widehat{\widetilde{U}}+U\widetilde{U}\widehat{\widetilde{U}}+U\widehat{U}\widetilde{U}+\widehat{U}\widetilde{U}\widehat{\widetilde{U}} \in\cl^0_{p,q},\qquad \Det(U)=-C_{(4)}=U\widehat{U}\widetilde{U}\widehat{\widetilde{U}}\in\cl^0_{p,q},\\
&&\Adj(U)=\widehat{U}\widetilde{U}\widehat{\widetilde{U}},\qquad U^{-1}=\frac{\widehat{U}\widetilde{U}\widehat{\widetilde{U}}}{\Det(U)}.
\end{eqnarray*}
In the case $n=4$, we have
\begin{eqnarray*}
&&C_{(1)}=U+\widehat{\widetilde{U}}+\widehat{U}^\va+\widetilde{U}^\va\in\cl^0_{p,q},\qquad C_{(2)}=-(U\widehat{\widetilde{U}}+U\widehat{U}^\va+U\widetilde{U}^\va+\widehat{\widetilde{U}}\widehat{U}^\va+ \widehat{\widetilde{U}}\widetilde{U}^\va+(\widehat{U}\widetilde{U})^\va)\in\cl^0_{p,q},\\
&&C_{(3)}=U \widehat{\widetilde{U}} \widehat{U}^\va+ U \widehat{\widetilde{U}} \widetilde{U}^\va+ U(\widehat{U} \widetilde{U})^\va+ \widehat{\widetilde{U}} (\widehat{U} \widetilde{U})^\va \in\cl^0_{p,q},\qquad \Det(U)=-C_{(4)}=U\widehat{\widetilde{U}}(\widehat{U} \widetilde{U})^{\va}\in\cl^0_{p,q},\\
&&\Adj(U)=\widehat{\widetilde{U}}(\widehat{U} \widetilde{U})^{\va},\qquad U^{-1}=\frac{\widehat{\widetilde{U}}(\widehat{U} \widetilde{U})^{\va}}{\Det(U)}.
\end{eqnarray*}
\end{theorem}

Note that in the formulas above we present only one of the possible realizations of the elements $C_{(k)}$, $k=1, \ldots, N$, and $\Adj(U)$. We can use different realizations of the trace $C_{(1)}$ (take different expressions $C_{(1)}= N \la U \ra_0$ from Theorem~\ref{th10}), determinant, adjugate (take different expressions $\Det(U)=N(U)$ and $\Adj(U)=F(U)$ from Theorem~\ref{thNF}), and other characteristic polynomial coefficients (for example, we can use the properties from Lemma \ref{lemC} to obtain other realizations).

\begin{proof} In the case $n=1$, we have $N=2$. Using (\ref{Ck}) and (\ref{r1}), we get
$$
 U_{(1)}=U,\qquad C_{(1)}=2 \la U \ra_0=U+\widehat{U},\qquad \Det(U)=-U_{(2)}=U(U+\widehat{U}-U)=U\widehat{U}.
$$
In the case $n=2$, we have $N=2$. Using (\ref{Ck}) and (\ref{r2}), we get
$$
 U_{(1)}=U,\qquad C_{(1)}=2 \la U \ra_0=U+\widehat{\widetilde{U}},\qquad \Det(U)=-U_{(2)}=U(U+\widehat{\widetilde{U}}-U)=U\widehat{\widetilde{U}}.
$$
In the case $n=3$, we have $N=4$. Using (\ref{Ck}) and (\ref{r3}), we get
\begin{eqnarray*}
&&U_{(1)}=U,\qquad C_{(1)}=4\la U \ra_0=U+\widehat{U}+\widetilde{U}+\widehat{\widetilde{U}},\qquad U_{(2)}=U(U-C_{(1)})=-U(\widehat{U}+\widetilde{U}+\widehat{\widetilde{U}}),\\
&&C_{(2)}=2\la U_{(2)} \ra_0=-2 \la U(\widehat{U}+\widetilde{U}+\widehat{\widetilde{U}})\ra_0=
-\frac{1}{2}(U\widehat{U}+U\widetilde{U}+U\widehat{\widetilde{U}}+\widehat{U} U+ \widehat{U} \widehat{\widetilde{U}}+\widehat{U}\widetilde{U}+
\widehat{\widetilde{U}}\widetilde{U}+U\widetilde{U}+\widehat{U}\widetilde{U}\\
&&+\widetilde{U}\widehat{\widetilde{U}}+\widehat{U}\widehat{\widetilde{U}}+ U\widehat{\widetilde{U}})=-(U\widetilde{U}+U\widehat{U}+U\widehat{\widetilde{U}}+\widehat{U}\widehat{\widetilde{U}}+\widetilde{U}\widehat{\widetilde{U}} +\widehat{U}\widetilde{U}),
\end{eqnarray*}
where we used $\widehat{U} U+\widehat{\widetilde{U}}\widetilde{U}=U\widehat{U}+\widetilde{U}\widehat{\widetilde{U}}$, which is equivalent to $[\widehat{U}, U]=([\widehat{U}, U])\,\widetilde\over$. This formula follows from the following reasoning. We have $([\widehat{U}, U])\,\widehat\over=[U,\widehat{U}]=-[\widehat{U}, U]$, i.e. $[\widehat{U}, U]\in\cl^1_{p,q}\oplus\cl^3_{p,q}$. Also we have $\la [\widehat{U}, U] \ra_3=\la [\widehat{U}, U] \ra_{n}=0$ by Lemma \ref{lem1}. Thus $[\widehat{U}, U]\in\cl^1_{p,q}$ and $[\widehat{U}, U]=([\widehat{U}, U])\,\widetilde\over$.

Further,
\begin{eqnarray*}
&&U_{(3)}=U(U_{(2)}-C_{(2)})=U(\widehat{U}\widehat{\widetilde{U}}+\widetilde{U}\widehat{\widetilde{U}}+\widehat{U}\widetilde{U}),\\ &&C_{(3)}=\frac{4}{3} \la U_{(3)} \ra_0=\frac{1}{3}(U\widehat{U}\widehat{\widetilde{U}}+U\widetilde{U}\widehat{\widetilde{U}}+U\widehat{U}\widetilde{U} + \widehat{U} U \widetilde{U}+\widehat{U} \widehat{\widetilde{U}}\widetilde{U}+\widehat{U}U\widehat{\widetilde{U}}+\widehat{U}\widehat{\widetilde{U}}\widetilde{U}+ \widehat{U}U\widetilde{U}+U\widehat{\widetilde{U}}\widetilde{U},\\
&&+U\widetilde{U}\widehat{\widetilde{U}}+U\widehat{U}\widehat{\widetilde{U}}+\widehat{U}\widetilde{U}\widehat{\widetilde{U}})=\frac{1}{3}(2 U\widehat{U}\widehat{\widetilde{U}}+2U\widetilde{U}\widehat{\widetilde{U}}+U\widehat{U}\widetilde{U}+\widehat{U}\widetilde{U}\widehat{\widetilde{U}}+ U\widehat{\widetilde{U}}(\widehat{U}+\widetilde{U})+2\widehat{U}(U+\widehat{\widetilde{U}})\widetilde{U})\\
&&=\frac{1}{3}(2 U\widehat{U}\widehat{\widetilde{U}}+2U\widetilde{U}\widehat{\widetilde{U}}+U\widehat{U}\widetilde{U}+\widehat{U}\widetilde{U}\widehat{\widetilde{U}}+ U(\widehat{U}+\widetilde{U})\widehat{\widetilde{U}}+2(U+\widehat{\widetilde{U}})\widehat{U}\widetilde{U})= U\widehat{U}\widehat{\widetilde{U}}+U\widetilde{U}\widehat{\widetilde{U}}+U\widehat{U}\widetilde{U}+\widehat{U}\widetilde{U}\widehat{\widetilde{U}},
\end{eqnarray*}
where we use $\widehat{U}\widetilde{U}=\widetilde{U}\widehat{U}\in\cen(\cl_{p,q})$, $\widehat{\widetilde{U}}U=U\widehat{\widetilde{U}}\in\cen(\cl_{p,q})$, $U+\widehat{\widetilde{U}}\in\cen(\cl_{p,q})$, $\widetilde{U}+\widehat{U}\in\cen(\cl_{p,q})$.
Finally, we get
\begin{eqnarray*}
&&\Det(U)=-U_{(4)}=U(C_{(3)}-U_{(3)})=U\widehat{U}\widetilde{U}\widehat{\widetilde{U}}.
\end{eqnarray*}
In the case $n=4$, we have $N=4$. Using (\ref{Ck}) and (\ref{T1}), we get
\begin{eqnarray*}
&&U_{(1)}=U,\qquad C_{(1)}=4 \la U \ra_0=U+\widehat{\widetilde{U}}+\widehat{U}^\va+\widetilde{U}^\va,\qquad U_{(2)}=U(U-C_{(1)})=-U(\widehat{\widetilde{U}}+\widehat{U}^\va+\widetilde{U}^\va),\\
&&C_{(2)}=2 \la U_{(2)} \ra_0=-2 \la U(\widehat{\widetilde{U}}+\widehat{U}^\va+\widetilde{U}^\va) \ra_0=
-\frac{1}{2}(U\widehat{\widetilde{U}}+U\widehat{U}^\va+U\widetilde{U}^\va +U\widehat{\widetilde{U}}+ \widetilde{U}^\va \widehat{\widetilde{U}} +\widehat{U}^\va \widehat{\widetilde{U}}+(\widehat{U}\widetilde{U})^\va \\
&&+(\widehat{U} U^\va)^\va\!+(\widehat{U} \widehat{\widetilde{U}}^\va)^\va\! +(\widehat{U} \widetilde{U})^\va\!+(\widehat{\widetilde{U}}^\va \widetilde{U})^\va\!+(U^\va \widetilde{U})^\va)=
-(U\widehat{\widetilde{U}}+U\widehat{U}^\va\!+U\widetilde{U}^\va\!+\widehat{\widetilde{U}}\widehat{U}^\va\!+ \widehat{\widetilde{U}}\widetilde{U}^\va\!+(\widehat{U}\widetilde{U})^\va),\\
&&U_{(3)}=U(U_{(2)}-C_{(2)})=U(\widehat{\widetilde{U}}\widehat{U}^\va+ \widehat{\widetilde{U}}\widetilde{U}^\va+(\widehat{U}\widetilde{U})^\va),\qquad C_{(3)}=\frac{4}{3}\la U_{(3)} \ra_0=\frac{1}{3}(U \widehat{\widetilde{U}} \widehat{U}^\va+ U \widehat{\widetilde{U}} \widetilde{U}^\va\\
&&+ U(\widehat{U} \widetilde{U})^\va+\widetilde{U}^\va U \widehat{\widetilde{U}}+\widehat{U}^\va U \widehat{\widetilde{U}} +(\widehat{U} \widetilde{U})^\va \widehat{\widetilde{U}}+(\widehat{U} \widetilde{U} U^\va)^\va+(\widehat{U}\widetilde{U}\widehat{\widetilde{U}}^\va)^\va+ (\widehat{U} (U \widehat{\widetilde{U}})^\va)^\va+(\widehat{\widetilde{U}}^\va \widehat{U}\widetilde{U})^\va\\
&&+ (U^\va \widehat{U} \widetilde{U})^\va+((U \widehat{\widetilde{U}})^\va \widetilde{U})^\va)=
U \widehat{\widetilde{U}} \widehat{U}^\va+ U \widehat{\widetilde{U}} \widetilde{U}^\va+ U(\widehat{U} \widetilde{U})^\va+ \widehat{\widetilde{U}} (\widehat{U} \widetilde{U})^\va,
 \end{eqnarray*}
where we used two times computer calculations (in Wolfram Mathematica) to simplify the expressions for $C_{(2)}$ and $C_{(3)}$,
because of nontrivial properties of the operation $\va$. Finally, we get
\begin{eqnarray*}
&&\Det(U)=-U_{(4)}=U(C_{(3)}-U_{(3)})=U\widehat{\widetilde{U}}(\widehat{U} \widetilde{U})^{\va}.
\end{eqnarray*}
\end{proof}

Note that in the case $n=5$, we can analogously take the formula (\ref{T1}) from Theorem \ref{th10}
and get explicit formulas for all characteristic polynomial coefficients $C_{(1)}, \ldots C_{(8)}$ using (\ref{Ck}):
\begin{eqnarray*}
&&\!\!U_{(1)}=U,\quad C_{(1)}=8 \la U \ra_0=2(U+\widehat{\widetilde{U}}+\widehat{U}^\va+\widetilde{U}^\va),\quad
U_{(2)}=U(U-C_{(1)})=-U(U+2\widehat{\widetilde{U}}+2\widehat{U}^\va+2\widetilde{U}^\va),\\
&&\!\!C_{(2)}=4 \la U_{(2)} \ra_0,\quad U_{(3)}=U(U_{(2)}-C_{(2)}),\quad C_{(3)}=\frac{8}{3} \la U_{(3)} \ra_0,\quad \ldots,\quad \Det(U)=-C_{(8)}=U(C_{(7)}-U_{(7)}).
\end{eqnarray*}
The final explicit formula for $\Det(U)$ should coincide after cumbersome calculations with the formula for the functional $N(U)$ from Theorem \ref{thNF} (this follows from the results of \cite{acus} using computer calculations). Similarly, from the results of \cite{acus}, it follows that the expressions (\ref{nu6}) and (\ref{nu62}) coincide with the determinant $\Det(U)$ in the case $n=6$.

From our results (see Theorems \ref{th10} and \ref{thLF}), it follows that all characteristic polynomial coefficients can be represented using only the operations of multiplication, summation, and the operations of conjugation $\va_1, \va_2, \ldots, \va_m$, where $m=[\log_2 n]+1$, in the case of arbitrary $n$. We have the recursive formulas (\ref{Ck}) and explicit formulas (\ref{Sk2}) for all $C_{(k)}$ in the case of arbitrary $n$. We see from the examples for small dimensions that the recursive formulas can be simplified using the properties of the operations $\va_1$, \ldots, $\va_m$ (and do this in the cases $n\leq 4$, see Theorem \ref{LFA}). Analytic simplification of these formulas in the case of arbitrary $n$ is non-trivial because of the non-trivial properties of the operations $\va_3$, \ldots $\va_m$ and seems to be an interesting task for further research. However, we can use the formulas (\ref{Ck}) without simplification.

The alternative way is to use the formulas for the determinant, other characteristic polynomial coefficients, and inverse using the operation $\stackrel{\over\quad}{\quad}$  (\ref{x2}) instead of the operations $\va_3$, \ldots, $\va_m$ (examples are given in Section \ref{sec:3}: the formulas (\ref{htb}) for the cases $n\leq 5$ and the formula (\ref{nu62}) for the case $n=6$) or instead of all operations $\va_1$, $\va_2$, $\va_3$, \ldots, $\va_m$ (we give several examples for the cases $n\leq 4$ after Theorem \ref{thLF}, in the general case we use the formulas (\ref{Sk}), (\ref{Sk2}), and (\ref{x2})).

\section{Conclusions}
\label{sec:5}

In this paper, we solve the problem of computing the inverse in Clifford algebras of arbitrary dimension. We present basis-free formulas for the trace, determinant, other characteristic polynomial coefficients, adjugate, and inverse in the real Clifford algebra $\cl_{p,q}$ for arbitrary $n=p+q$. These formulas do not use matrix representations and use only operations in Clifford algebras.

The formulas of the first type (\ref{Ck}) are recursive, the formulas of the second type (\ref{Sk2}) are explicit and use the complete Bell polynomials. The formulas of both types use the operations of multiplication, summation, and the operation $\la \quad \ra_0$ of projection onto the subspace of grade $0$. The operation $\la \quad \ra_0$ can be realized using one operation of conjugation $\stackrel{\over\quad}{\quad}$ (\ref{x2}), or using $m=[\log_2 n]+1$ operations of conjugation $\va_1, \ldots, \va_m$ (\ref{ocst}). Sometimes these formulas can be simplified using the properties of the operations $\va_1, \ldots, \va_m$ (see the discussion at the end of the previous section). We present simplification of these formulas in the case of small dimensions (see Theorem \ref{LFA}). Analytic simplification of the formulas in the case of arbitrary $n$ is non-trivial and seems to be an interesting task for further research. However, we can use the formulas (\ref{Ck}), (\ref{det01}), (\ref{Sk2}), and (\ref{det02}) without any simplification. We can use different formulas for different purposes.

The formulas (\ref{det01}) and (\ref{det02}) can be interpreted as definitions of the concept of determinant in Clifford algebra. They are equivalent to the definition (\ref{det}) but do not use matrix representations. The formulas (\ref{Ck}) and (\ref{Sk2}) can be interpreted as two different (equivalent) definitions of characteristic polynomial coefficients in Clifford algebras. We see that the condition of invertibility of Clifford algebra element ($\Det(U)\neq 0$) does not depend on $p$ and $q$ in the real Clifford algebras $\cl_{p,q}$ for arbitrary $n$. The recursive and explicit formulas from Theorems \ref{thLF} and \ref{LFA} can be used in symbolic computation.

We use the results of this paper to obtain basis-free solution to the Sylvester equation of the form $AX+XB=C$ for known $A, B, C\in\cl_{p,q}$ and unknown $X\in\cl_{p,q}$ and its particular case, the Lyapunov equation (with $B=A^H$) \cite{ShirokovSylv}. These equations are widely used in image
processing, control theory, stability analysis, signal processing, model reduction, and many more.

The main results of this paper remain true for the complexified Clifford algebras $\C\otimes\cl_{p,q}$. All results of Sections \ref{sec:2} and \ref{sec:3} are generalized to the case of $\C\otimes\cl_{p,q}$ without any changes. In Section \ref{sec:4}, we should take the matrix representation $\beta$ (\ref{beta}) instead of (\ref{beta2}) in all considerations. The characteristic polynomial coefficients $C_{(k)}$ (in particular, the trace and the determinant) will be complex numbers. All formulas of Theorems \ref{thLF} and \ref{LFA} will be valid in the case of $\C\otimes\cl_{p,q}$.

The real Clifford algebras $\cl_{p,q}$ are isomorphic to the matrix algebras over $\R$, $\C$, $\R\oplus\R$, $\H$, or $\H\oplus\H$ depending on $p-q\mod 8$, the complexified Clifford algebras $\C\otimes\cl_{p,q}$ are isomorphic to the matrix algebras over $\C$ or $\C\oplus\C$ depending on $n\mod 2$. The advantage of Clifford algebras over matrix algebras is a more powerful mathematical apparatus, which allows us to naturally realize different geometric structure, spin group, spinors, etc. At the same time, the matrix methods are also useful for different purposes and applications. Therefore the problem arises to transfer the matrix methods to the formalism of Clifford algebras. Note the papers \cite{AblMP, Helm, polar, MM, polar2}. An interesting task is to generalize results of this paper to the Moore-Penrose inverse (pseudo-inverse, \cite{Golub}), which is widely used in computer science and engineering.

\section*{Acknowledgment}

The author is grateful to N.~Marchuk and N.~Khlyustova for useful discussions. The author is grateful to the anonymous reviewers for their careful reading of the paper and helpful comments on how to improve the presentation. The results of this paper were reported at the 9th International Conference on Mathematical Modeling (Yakutsk, 2020), the 12th International Conference on Clifford Algebras and Their Applications in Mathematical Physics (Hefei, 2020), and the International Conference ``Computer Graphics International'' (Geneva, 2020, within the workshop ``Empowering Novel Geometric Algebra for Graphics and Engineering''). The author is grateful to the organizers and the participants of these conferences for fruitful discussions.

The publication was prepared within the framework of the Academic Fund Program at the HSE University in 2020–2021 (grant 20-01-003).

\end{document}